\documentclass{article}
\usepackage[T1]{fontenc}
\usepackage[latin9]{inputenc}
\usepackage{amsmath}
\usepackage{amssymb}
\usepackage{mathtools}
\usepackage{mathtools}
\usepackage{amsthm}
\usepackage[linesnumbered,ruled,vlined]{algorithm2e}
\usepackage{graphicx}
\usepackage{pgfplots}
\pgfplotsset{compat=1.18}
\usepackage{subcaption}
\usepackage{amsmath}
\usepackage{appendix}
\usepackage[colorlinks=true,citecolor=blue,linkcolor=blue,urlcolor=black]{hyperref}
\usepackage{cleveref}
\usepackage{xcolor}
\usepackage{comment}
\usepackage{thm-restate}
\usepackage{natbib}
\usepackage{fullpage}
\usepackage{array} 
\usepackage{tikz}
\usepackage{caption}
\usepackage{enumitem}

\newtheorem{theorem}{Theorem}[section]
\newtheorem{theoremnatural}{Theorem}
\newtheorem{proposition}[theorem]{Proposition}
\newtheorem{lemma}[theorem]{Lemma}

\newtheorem{corollary}[theorem]{Corollary}
\newtheorem{definition}[theorem]{Definition}
\newtheorem{remark}[theorem]{Remark}
\newtheorem{observation}[theorem]{Observation}

\newtheorem{conjecture}[theorem]{Conjecture}
\newtheorem*{theorem*}{Theorem}

\allowdisplaybreaks

\newcommand\blfootnote[1]{
  \begingroup
  \renewcommand\thefootnote{}
  \NoHyper\footnote{#1}\endNoHyper
  \addtocounter{footnote}{-1}
  \endgroup
}

\newcommand{\reals}{\mathbb{R}}

\DeclarePairedDelimiter{\floor}{\lfloor}{\rfloor}
\DeclarePairedDelimiter{\ceil}{\lceil}{\rceil}
\newcommand{\eps}{\varepsilon}

\newcommand{\cmmnt}[1]{}
\newcommand{\lightagents}{L}
\newcommand{\maxlightreward}{\textsc{Max-Reward-Light}}
\newcommand{\agents}{A}
\newcommand{\instance}{\langle \agents,f,c \rangle}
\newcommand{\goodobj}{BEST}

\newcommand{\MAX}{\textsc{Max-}}

\newcommand{\maxreward}{\textsc{Max-Reward}}
\newcommand{\maxrevenue}{\textsc{Max-Profit}}
\newcommand{\maxwelfare}{\textsc{Max-Welfare}}

\newcommand{\myheight}{8.6cm}

\newcommand{\gaprewardxos}{\textsc{PoF-Reward-XOS}}

\newcommand{\gapphixos}{\textsc{PoF-$\varphi$-XOS}}

\newcommand{\gaprewardsubmodular}{\textsc{PoF-Reward-Submodular}}
\newcommand{\gaprevenuesubmodular}{\textsc{PoF-Profit-Submodular}}
\newcommand{\gapwelfaresubmodular}{\textsc{PoF-Welfare-Submodular}}
\newcommand{\gapphisubmodular}{\textsc{PoF-$\varphi$-Submodular}}

\newcommand{\gapphiadditive}{\textsc{PoF-$\varphi$-Additive}}

\newcommand{\gapreward}{\textsc{PoF-Reward}}
\newcommand{\gaprevenue}{\textsc{PoF-Profit}}
\newcommand{\gapwelfare}{\textsc{PoF-Welfare}}
\newcommand{\gapphi}{\textsc{PoF-$\varphi$}}

\newcommand{\contract}{\vec{\alpha}}
\newcommand{\nash}{\mathsf{NE}}

\title{Budget-Feasible Contracts
}
\author{ 
Michal Feldman$^\ast$
\quad
Yoav Gal-Tzur$^\dagger$
\quad 
Tomasz Ponitka$^\ddagger$
\quad
Maya Schlesinger$^\mathsection$
}
\date{\today}

\begin{document}

\maketitle

\blfootnote{
This is the full version of an extended abstract that will appear in the proceedings of the 26th ACM Conference on Economics and Computation (EC'25). 
An earlier version of this arXiv article included additional results on the multi-agent setting with combinatorial actions, which are not part of the EC'25 extended abstract.

This project has been partially funded by the European Research Council (ERC) under the European Union's Horizon 2020 research and innovation program (grant agreement No. 866132), by an Amazon Research Award, by the Israel Science Foundation Breakthrough Program (grant No.~2600/24), and by a grant from TAU Center for AI and Data Science (TAD).}
\blfootnote{$^\ast$Tel Aviv University and Microsoft ILDC, Israel. Email: \texttt{mfeldman@tauex.tau.ac.il}}
\blfootnote{$^\dagger$Tel Aviv University, Israel. Email: \texttt{yoavgaltzur@mail.tau.ac.il}}
\blfootnote{$^\ddagger$Tel Aviv University, Israel. Email: \texttt{tomaszp@mail.tau.ac.il}}
\blfootnote{$^\mathsection$Tel Aviv University, Israel. Email: \texttt{mayas1@mail.tau.ac.il}}

\begin{abstract}

The problem of computing near-optimal contracts in combinatorial settings has recently attracted significant interest in the computer science community. Previous work has provided a rich body of structural and algorithmic insights into this problem. However, most of these results rely on the assumption that the principal has an unlimited budget for incentivizing agents, an assumption that is often unrealistic in practice. This motivates the study of the optimal contract problem under budget constraints.

In this work, we study multi-agent contracts with binary actions under budget constraints. Our contribution is threefold. First, we show that all previously known approximation guarantees on the principal's utility extend (asymptotically) to budgeted settings. Second, through the lens of budget constraints, we uncover insightful connections between the standard objective of maximizing the principal's utility and other relevant objectives. Specifically, we identify a broad class of objectives, which we term BEST {(BEyond STandard)} objectives, including reward, social welfare, and principal's utility, and show that they are all equivalent (up to a constant factor), leading to approximation guarantees for all BEST objectives. Third, we introduce the price of frugality, which quantifies the loss due to budget constraints, and establish near-tight bounds on this measure, providing deeper insights into the tradeoffs between budgets and incentives.

\end{abstract}

\newpage
\section{Introduction}

Contract design is a fundamental topic in microeconomic theory, highlighted by the 2016 Nobel Prize in Economics awarded to Hart and Holmstr\"{o}m. At the heart of contract theory lies the principal-agent model, where a principal delegates a costly {project} to agents and incentivizes them through a contract that specifies payments based on observed outcomes.
With the growth of online markets for services, the field has recently attracted significant attention from the computer science community; for a recent survey, see  \cite{DuttingFT24}.

A natural example where contract theory plays a crucial role is in crowdsourced data annotation. Consider a company, hospital, or researcher seeking to label large datasets using online platforms such as Amazon Mechanical Turk. 
In this setting, multiple annotators are hired, but their contributions may exhibit decreasing marginal returns.
For instance, the first few annotators provide valuable insights, but as more are added, their additional input becomes less informative. 
The principal (the entity commissioning the work) faces the challenge of designing a payment scheme that incentivizes high-quality effort despite being unable to directly monitor the agent's actions.
Contract theory provides the tools to design such incentives, ensuring that annotators are compensated in a way that aligns their efforts with the principal's objectives.

The first model of incentivizing \emph{teamwork}, capturing the preceding scenario, was introduced by \cite{holmstrom1982moral}. More recently,
\cite{babaioff2006combinatorial,BabaioffFNW12} introduced a combinatorial model for multi-agent contracts, which was later generalized by \cite{duetting2022multi}. In these settings, the outcome depends on the delicate interplay between the efforts of  different agents, captured by a set function that assigns each subset of agents the principal's expected reward. This combinatorial structure poses significant algorithmic and computational challenges.

More concretely, in the model considered in  \cite{duetting2022multi}, a principal delegates the execution of a {project} to a set of agents $\agents$. 
The project has a binary outcome: it either succeeds or fails.
Moreover, each agent $i \in \agents$ has a binary action:
either exert effort, at cost $c_i$, or not {(at no cost)}.
Crucially, the effort of the agent is hidden from the principal.
We are given a set function $f: 2^{\agents} \to [0,1]$, where $f(T)$ denotes the probability that the project succeeds when all agents in $T \subseteq \agents$ exert effort.
The principal's reward upon success is normalized to $1$; thus the function $f$ is also the principal's expected reward.
The principal incentivizes the agents to exert effort using a contract that specifies the payment to each agent upon the project's success. 
Each agent aims to maximize their utility, which is the expected payment from a given contract\footnote{The expected payment to an agent may depend on the efforts of other agents, via the function $f$.} minus the cost of the agent's chosen action.
Given a contract, the agents have been put in a game, and play a Nash equilibrium. The payment required in order to incentivize {the agents} $T \subseteq \agents$ to exert effort is denoted by $p(T)$ (and is given by a closed-form expression).
Since the principal's reward is normalized to $1$,  $p(T) \in [0,1]$ is precisely the fraction of the reward transferred to the agents.

The goal is to find the optimal contract, commonly defined as the contract maximizing
the principal's expected utility, given by $(1-p(T)) \cdot f(T)$. Notably, the problem of computing an optimal contract reduces to finding the optimal team to incentivize. 
The main contribution of \cite{duetting2022multi} is establishing constant-factor approximation guarantees for a broad class of reward functions $f$, including submodular and XOS functions, {using value and demand oracle access, respectively (see Section~\ref{sec:model} for definitions)}.

\vspace{0.1cm}
\noindent \textbf{Budget Constraints.}
A significant limitation of the result in \cite{duetting2022multi} is that it assumes the principal has an unlimited budget for incentivizing agents.
However, in practice, principals often face budget constraints that limit their ability to offer incentives.
Consider the data annotation example discussed earlier: a company, hospital, or researcher typically has a fixed budget allocated for labeling datasets. 
This motivates the study of the optimal contract problem under budget constraints, i.e., the problem of maximizing $(1-p(T)) \cdot f(T)$ subject to the constraint $p(T)\le B$ for a given budget $B\in (0,1]$.
Notably, budget constraints have been studied in multi-agent contract design but in a non-combinatorial setting \cite{hann2024optimality} (see \Cref{subsec:RelatedWork} for details).

Our first question is whether the constant-factor approximation guarantees for the principal's utility in non-budgeted multi-agent combinatorial settings, as established by \cite{duetting2022multi}, extend to budgeted settings.

\vspace{0.1cm}
   \noindent  \textbf{Question 1:} \emph{Can we efficiently find near-optimal budget-feasible contracts in multi-agent settings?}
   \vspace{0.1cm}

The result in \cite{duetting2022multi}, like much of the combinatorial contracts literature, focuses on maximizing the principal's utility, which we refer to as profit. However, in real-world settings, other objectives may be just as relevant. For instance, one might aim to maximize the project's 
success probability (reward), or seek to optimize social welfare (defined as the difference between expected reward and total cost).
Motivated by this, our second goal
is to examine alternative objectives, including (but not restricted to) expected reward and social welfare.\footnote{In a concurrent work, \cite{aharoni2025welfare} consider social welfare and reward maximization in the same model as ours --- the one introduced in \cite{duetting2022multi}.
Reward maximization has been considered in a non-combinatorial budgeted settings by \cite{hann2024optimality}.}

   \vspace{0.1cm}
    \noindent \textbf{Question 2:} \emph{Can natural objectives beyond profit also be efficiently approximated?}
        \vspace{0.1cm}
        
Additionally, budget constraints can naturally limit profit as well as other desired objectives. Our next objective is to examine the tradeoffs between budget and incentives by quantifying the resulting loss.

        \vspace{0.1cm}
     \noindent \textbf{Question 3:} \emph{{What is the worst-case loss incurred due to budget constraints?}}
        \vspace{0.1cm}

In this work we provide answers the above three questions.

\subsection{Our Results}

Our computational results are summarized in \Cref{tab:binary_actions}.

\begin{table}[t]
    \vspace{-0.5cm}
    \centering
    \renewcommand{\arraystretch}{1.3}
    \begin{tabular}{|>{\centering\arraybackslash}p{3.5cm}|>{\centering\arraybackslash}p{5.5cm}|>{\centering\arraybackslash}p{6cm}|}
    \hline
     & \textbf{Maximizing Profit} & \textbf{Maximizing any BEST Objective} \\ 
     & \textbf{Without Budgets (Prior Work)} & \textbf{With Budgets (New)} \\ 
    \hline
    \textbf{Additive} & FPTAS & FPTAS$^{*}$ \\ 
     & {\footnotesize \cite{duetting2022multi}} & {\footnotesize (\Cref{prop:FPTAS} and \Cref{remark:additive_f_sw_reward_fptas})}  \\ 
     \hline
    \textbf{Submodular} & $O(1)$-approx. (value queries) & $O(1)$-approx. (value queries) \\ 
     & {\footnotesize \cite{duetting2022multi}} & {\footnotesize (\Cref{Cor:constant_approx})}  \\ \hline
    \textbf{XOS} & $O(1)$-approx. (demand queries) & $O(1)$-approx. (demand queries) \\ 
      & {\footnotesize \cite{duetting2022multi}} & {\footnotesize (\Cref{Cor:constant_approx})}  \\ \hline
    \end{tabular}
    \caption{
    The computational complexity of finding near-optimal contracts in the multi-agent model of \cite{duetting2022multi}. The results in the left column are from \cite{duetting2022multi}, while the right column presents our new results. 
    (*) Our results for submodular and XOS functions apply to all \goodobj\ objectives (including profit, welfare, and reward), whereas the FPTAS for additive functions applies only to profit, welfare, and reward.
    }
    \label{tab:binary_actions}

\end{table}

\paragraph{Approximation Guarantees for Multi-Agent Contracts.}
Our main result is an efficient algorithm that achieves a constant-factor approximation to the principal's profit under budget constraints, {when the function $f$ is XOS}. Moreover, we generalize this result to additional natural objectives beyond profit, which we term \goodobj\ (BEyond STandard) objectives. We formally define the class of BEST objectives in \Cref{sec:good-objectives}.
This class includes profit (the principal's utility), social welfare, reward, and any convex combination of these. 
{Our main result is cast in the following theorem.}

\begin{theoremnatural}[Constant-Factor Approximation to \goodobj\  Objectives under Budget]\label{thm:ConstApproxBEST}
For any \goodobj\ objective $\varphi$ (including profit, reward, and welfare), when $f$ is XOS, there exists an algorithm that gives a constant-factor approximation to the optimal {budget-feasible} contract with respect to $\varphi$.
Furthermore, this algorithm runs in polynomial time, under demand oracle access to $f$.
The same guarantees hold for submodular $f$ functions, using only value oracle access to $f$.
\end{theoremnatural}

In fact, we prove a significantly stronger result, establishing an essential equivalence (up to a constant factor) between any two \goodobj\ objectives, under any two budgets. 
Specifically, we show that obtaining a constant-factor approximation to \emph{some} \goodobj\ objective under \emph{some} budget,  implies a constant-factor approximation to \emph{any} \goodobj\ objective, under \emph{any} budget.

\begin{theoremnatural}[All Objectives Are Equivalent (Informal)]
    When $f$ is XOS, for any \goodobj\ objectives $\varphi, \varphi'$ and any budgets $B,B' \in (0,1]$, the problem of approximately maximizing $\varphi$ under budget $B$ reduces to the problem of approximately maximizing $\varphi'$ under budget $B'$.
\end{theoremnatural}

By observing that the setting of \cite{duetting2022multi} admits an implicit budget of $1$, this equivalence, combined with the constant-factor approximation to profit given in \cite{duetting2022multi}, implies Theorem~\ref{thm:ConstApproxBEST}. 
Moreover, this equivalence implies that the hardness result for profit maximization given by \cite{ezra2023Inapproximability}
extends to all \goodobj\ objectives and budgets.
In particular, no constant-factor approximation is possible when $f$ is XOS with just value oracle access. 
Hence, the demand query assumption for XOS $f$ in \Cref{thm:ConstApproxBEST} is essential and cannot be replaced by value oracle access alone.

In addition, 
we obtain stronger guarantees for the  special case of additive 
$f$.
In particular, we extend the FPTAS for profit from \cite{duetting2022multi} to budgeted settings (see Appendix~\ref{app:additive_fptas}). Furthermore, we observe that computing (near-)optimal budget-feasible contracts for both social welfare and expected reward reduces to the \textsc{Knapsack} problem, allowing us to derive an FPTAS for these objectives as well (see \Cref{remark:additive_f_sw_reward_fptas}).

\paragraph{Price of Frugality.}

Our next set of results addresses Question 3 by quantifying the loss incurred due to budget constraints. To this end, we introduce a notion of the {\em Price of Frugality} (PoF). 
Given an objective $\varphi$ and two budgets $b < B$, the PoF is defined as the worst-case ratio (over all instances where every individual agent is incentivizable under budget $b$) between the maximum value of $\varphi$ under budget $B$ and that under budget $b$. Our results hold for all \goodobj\ objectives (including profit, expected reward, and social welfare).

\begin{theoremnatural}[Near-Tight Bounds on Price of Frugality]
    When $f$ is XOS, for any \goodobj\  objective $\varphi$ and budgets $b \leq B \leq 1$, the price of frugality 
    is $\Theta(\min(B/b,n))$.
\end{theoremnatural}

One way to interpret the bound above is that the value of any \goodobj\ objective decreases at most linearly as the budget $b$ decreases, provided that $b$ is sufficient to incentivize any singleton. In particular, reducing the budget by a constant factor results in only a constant-factor loss in the objective.

While this bound is asymptotically tight, we also characterize the exact price of frugality for reward and welfare when $f$ is submodular (\Cref{thm:main_price}). 
Additionally, we prove two separation results: (1) the price of frugality for reward is strictly higher under XOS $f$ than under submodular $f$ (see \Cref{lem:hard_pof_xos}), and (2) under subadditive $f$, the price of frugality for reward is $\Omega(\sqrt{n})$, where $n$ is the number of agents (see \Cref{lem:pof_subadditive}).

\subsection{Key Techniques and Insights}

\paragraph{Downsizing Lemmas.} 
In budgeted settings, a natural question arises: Can any given team (potentially exceeding the budget) be reduced to the ``most effective'' agents to satisfy the budget constraint while preserving a guarantee on the expected reward? We answer this question through our downsizing lemmas, demonstrating that payments can be scaled down to almost any target---potentially to the cost of incentivizing a single agent---while ensuring that the expected reward decreases at most linearly with the payment. Furthermore, this procedure can be executed in polynomial time with value oracle access.

We establish downsizing lemmas for  submodular and XOS $f$ using a bag-filling approach.
For the submodular case, 
we iteratively add agents to the bag as long as (an upper bound of) the payment remains within the target budget, then we either return the set, if it meets the efficiency threshold, or start a new bag otherwise.
We use the monotonicity of marginal values of submodular $f$ to upper bound the resulting payment.
For XOS $f$, where marginal values are not necessarily monotone, the situation becomes more challenging. 
In particular, the expression used to determine when to terminate a bag in the algorithm described above no longer serves as an upper bound on the payments.
Hence, running this algorithm might
potentially lead to a payment that exceeds the target amount.
To bring the payment below the target, we establish a scaling lemma for XOS functions, akin to \cite[Lemma 3.5]{duetting2022multi}, which selectively removes agents to ensure that the remaining ones have sufficiently high marginal values. This allows us to upper bound the required payment, while preserving sufficiently high expected reward.

The downsizing lemmas serve as a key component of our algorithmic reductions. They are also closely related to our analysis of the price of frugality. In particular, our lower bounds on the price of frugality imply that the downsizing algorithm for submodular functions is tight for any given target payment, while the algorithm for XOS functions is tight up to a constant factor. 

\paragraph{Light Agents.} 
A key idea in the analysis of budget-feasible contracts is to distinguish between \emph{light} and \emph{heavy} agents. Light agents are those who can be incentivized to exert effort with a payment of at most $1/2$, and heavy agents are those that require a payment higher than $1/2$.

When all agents are light, {we make} two key observations. First, if the total payment for a given team is at most $1/2$, then the profit {it generates} is at most a constant factor away from {its} reward.
Second, our downsizing lemmas allow us to reduce a budget-feasible team of light agents to a team whose total payment is at most $1/2$, while only losing a constant factor of the reward.
Together, these observations imply that {approximately} maximizing profit {is equivalent} to {approximately} maximizing the reward.
This argument applies not only for profit, but for any objective that lies between profit and reward (e.g., social welfare).

In addition, we observe that for reward maximization, any two budgets are equivalent.
Indeed, for any two budgets $B,B'\in (0,1]$, one can rescale the costs\footnote{Note that, in the new scaled instance some agents may no longer be light. For ease of exposition, we ignore this subtlety here. A complete analysis is given in \Cref{sec:structural_insights}.} by a factor ${B'}/{B}$. 
In the scaled instance, a budget constraint of $B'$ is equivalent to a budget constraint of $B$ in the original instance. 
Since reward maximization is independent of costs, this implies that any approximation algorithm under budget $B$ can be adjusted through simple scaling to give the same approximation under budget $B'$.

To summarize, the above analysis shows that when all agents are light: (1) For a given budget constraint, any two objectives that lie between profit and reward are equivalent up to a constant factor. (2) For the problem of maximizing the principal's reward, any two budgets are equivalent. Together, these results imply that {any two such (objective, budget) pairs are equivalent up to a constant factor}.

In the presence of heavy agents the above analysis no longer holds.
In particular, the team that (approximately) maximizes reward may have no relevance for (approximate) profit maximization.
For example, suppose $B=1$, and there exists an agent $i\in \agents$ such that $f(\{i\})$ is significantly larger than $f(S)$ for any $S\subseteq \agents \setminus \{i\}$, and the payment for incentivizing $i$ is $p(\{i\})=1$. 
Clearly, the singleton $\{i\}$ is the only approximately optimal set with respect to reward. At the same time, any set that contains $i$ generates profit at most $0$, illustrating the disconnect between the two objectives. 
To handle scenarios with heavy agents, we introduce the notion of \goodobj\ objectives, as explained below.

\paragraph{\goodobj\ Objectives.}

For general settings (which include both light and heavy agents), we observe that any budget feasible set of agents can contain at most one heavy agent.
This leads us to identify a {key} property of objective functions that 
{enables separate handling of heavy and light agents}.
Roughly speaking, such objectives can be approximated by {selecting} the better of the two: (1) a good set of light agents, {or} (2) the best heavy agent. 
{As finding the best heavy agent is {straightforward},} 
{obtaining} an approximately optimal solution for the general case effectively reduces
{to solving the problem in the setting with only light agents.}
We call objectives that satisfy {the corresponding} property, and {are additionally} sandwiched between profit and reward, \goodobj\ objectives. {We show that} profit, social welfare, reward, and any convex combination of the three qualify as \goodobj\ objectives.

\subsection{Organization}
The organization of the paper is as follows: In \Cref{sec:model}, we define the multi-agent model with budget constraints and provide necessary preliminaries. In \Cref{sec:structural_insights}, we introduce key technical insights, including the downsizing lemmas and the central notions of {light agents} and {BEST objectives}. 
In \Cref{sec:computational_results}, we give our main algorithmic results, including a constant-factor approximation for the optimal budget-feasible contract when $f$ is XOS, and an essential  equivalence between any two  BEST objectives and any two budgets.
We introduce the {Price of Frugality (PoF)} in \Cref{sec:pof}, and give asymptotically tight bounds for PoF for any BEST objective in XOS instances {and tight bounds for PoF for reward and welfare in submodular instances}.

\subsection{Related Work}\label{subsec:RelatedWork}

\paragraph{Combinatorial Contracts.}

A combinatorial model for contracting multiple agents with binary effort was introduced by \cite{babaioff2006combinatorial,BabaioffFNW12}, where the principal seeks to incentivize the optimal set of agents. They focused on the case where the function $f$, mapping agents' efforts to the principal's expected reward, is Boolean. Subsequent work \cite{babaioff2009free, babaioff2006mixed} studied free-riding and mixed strategies in this setting.
The work of \cite{duetting2022multi} generalized this model to set functions from the complements-free hierarchy of \cite{lehmann2001combinatorial}. They showed that when $f$ is XOS, the optimal contract admits a constant-factor approximation using demand and value queries to $f$. 
Building upon this model, \cite{gong2025approximating} consider the problem of incentivizing a team of at most $k$ agents (i.e., maximizing profit under cardinality constraints), and give a poly-time algorithm when $f$ is XOS, with demand query access.

In a concurrent work, \cite{aharoni2025welfare} consider the same multi-agent model as in \cite{duetting2022multi}, focusing primarily on maximizing social welfare. They present polynomial-time algorithms that achieve a constant-factor approximation to the optimal welfare when $f$ is XOS, {using} demand oracle access.
In the case of symmetric agents, they further improve the approximation ratio, using only logarithmically (in the number of agents) many queries. 
Some of their results, obtained for an implicit budget of $B=1$, are extended to any $0 < B \le 2 - \Theta(1)$.
They also provide bounds on the ratio between optimal welfare and profit.

In a follow-up work, \cite{multimulti} extend their results from \cite{duetting2022multi} to settings where agents can perform arbitrary combinations of a given set of actions, proving that when $f$ is submodular, the optimal contract can be approximated using demand queries.
A further relaxation of the binary-actions assumption appears in \cite{cacciamani2024multi}, where the (possibly exponentially large) contract instance is given as input.

A related line of work examines contracting multiple agents, each assigned an individual task with an observable outcome. Unlike the aforementioned combinatorial models, the principal contracts each of the agent based on her individual outcome.
It was established in \cite{castiglioni2023multi} that the optimal contract can be computed when the agents' individual outcomes exhibit increasing returns.

Moreover, \cite{dutting2022combinatorial} introduce a \emph{single-agent} model where the agent can perform any subset of $n$ costly actions, with a set function $f$ mapping each subset to the expected reward. They showed that when $f$ is gross substitutes, the optimal contract can be computed in polynomial time using value queries. Subsequent works \cite{deo2024supermodular, dutting2024query, ezra2023Inapproximability, contractsBeyondGS,feldman2025ultraefficient} further explore the tractability frontier of this model.

The above studies focus on binary outcome settings, where the optimal contract takes a linear form---that is, the principal pays the agent a {fixed} fraction of {her} reward. 
Additional studies have examined the properties of linear contracts. \cite{carroll2015robustness,peng2024optimal,dutting2019simple} established their max-min optimality under various forms of uncertainty.
\cite{dutting2019simple} also quantified the (worst-case) performance loss associated with using linear contracts.
Also in the single-agent case, \cite{dutting2021complexity} explored a setting with a combinatorial outcome space that can be succinctly represented. They introduce the notion of $\delta$-IC contracts, where the principal's preferred action is the agent's best-response up to an additive factor of $\delta$. They give an algorithm for finding the optimal $\delta$-IC contract whose running time is polynomial in $1/\delta$ when the number of actions is constant.

\paragraph{Contracts with Budget Constraints.}
Budget constraints typically introduce significant algorithmic challenges, a phenomenon well-studied in both classical domains, such as the \textsc{Knapsack} problem, and modern algorithmic domains, such as budget-feasible auctions \cite{Singer10}.
In the context of contract design, \cite{hann2024optimality} studied multi-agent contracts under budget constraints, aiming to maximize the principal's reward. Their setting differs from ours in that each agent performs an independent task with a binary outcome observable by the principal, leading to independent payments and no need for equilibrium analysis. Moreover, unlike the combinatorial setting, the principal's reward is a linear function of the agents' individual success probabilities.

\paragraph{Other Contractual Settings.}
The problem of contracting agents with hidden types has been studied in both single- and multi-agent settings \cite{alon2021contracts, alon2022bayesian, castiglioni2025reduction, CastiglioniM021, GuruganeshSW023, guruganesh2021contracts, castiglioni2022designing}.
The intersection of contracts and learning has also received significant attention \cite{ZhuBYWJJ23, DuettingGuruganeshSchneiderWang23, ho2014adaptive, BacchiocchiC0024,chen2024boundedcontractslearnableapproximately, duetting2025pseudodimensioncontracts}.
{Another recent research direction explores the design of ambiguous contracts \cite{DuttingFP23,dutting2024ambiguous,duetting2025succinct}}.
{For an extensive survey covering these and other emerging topics, see \cite{DuttingFT24}.}

\section{Model and Preliminaries} \label{sec:model}

\paragraph{The Model.} 
We focus on the multi-agent
model introduced by \cite{duetting2022multi}. In this model, a principal delegates the execution of a project and interacts with a set $\agents$ of $n$ agents. Each agent $i \in A$ chooses between two actions: exerting effort or not. Exerting effort incurs a cost of $c_i \geq 0$ for the agent, while not exerting effort has no cost.

We focus on the binary-outcome case\footnote{In fact, 
all of our results extend to the case of linear contracts in the model with multiple outcomes.} where a project can either succeed or fail.
A function $f: 2^A \rightarrow [0,1]$ maps each subset of agents who exert effort to the project's success probability.
If the project succeeds, the principal receives a reward, which we normalize to $1$; otherwise, the reward is $0$.
We also refer to $f$ as the \emph{reward} function, as $f(S)$ is precisely the principal's expected reward when the set of agents $S$ exerts effort.
We denote an instance of the multi-agent
model as $\instance$, where $\agents$ is the set of agents, $f$ is the reward function, and $c=(c_i)_{i\in \agents}$ is the vector of agent costs.

Crucially, the principal cannot observe the actions of the agents, only whether or not the project succeeded. Therefore, to incentivize the agents to exert effort, 
the principal designs a contract $\contract = (\contract_1,\dots,\contract_n)$, where $\contract_i$ denotes the non-negative payment the principal transfers to agent $i$ if the project succeeds. 
In a binary-outcome setting, this form of (linear) contract is without loss of generality for a principal which tries to maximize her expected utility \cite{duetting2022multi}.

\paragraph{Utilities and Equilibria.}
For any contract $\contract=(\contract_1,\ldots,\contract_n)$, and a set $S \subseteq \agents$ of agents who exert effort, the principal's utility is her expected reward minus the expected payment to the agents, that is, 
$\left(1-\sum_{i\in \agents} \contract_i\right) \cdot f(S)$. 
Each agent's utility is the expected payment made to them by the principal, minus their cost if they exerted effort, i.e., $\contract_i \cdot f(S) - c_i$ if $i \in S$ and by $\contract_i \cdot f(S)$ otherwise. 
Importantly, the cost incurred by an agent depends only on whether they exerted effort, regardless of the project's outcome; and the payment depends only on the project's outcome, regardless of their effort.

Once the principal commits to a contract, the agents engage in a (pure) Nash equilibrium of the induced game. 
A contract $\contract$ is said to incentivize a set $S\subseteq A$ of agents to exert effort (in equilibrium) if 
\begin{align}\label{eq:NE_def}
\contract_i \cdot f(S) -  c_i  \ge \contract_i \cdot f(S \setminus \{i\}) & &&\text{for all }i\in S\text{, and} \\
\contract_i \cdot f(S) \ge \contract_i \cdot f(S\cup \{i\}) -  c_i &&& \text{for all }i\notin S. \nonumber
\end{align}
Note that a given contract may admit multiple equilibria. 
Whenever the equilibrium $S \subseteq \agents$ is induced by the contract $\contract$ we denote it with $S \in \nash(\contract)$.
As is standard in the literature, we assume tie-breaking favors the principal, allowing the principal to select the optimal Nash equilibrium, unless mentioned otherwise. 

\paragraph{The Contract Design Problem.}
As implied by \Cref{eq:NE_def}, in order to incentivize the agents of a set $S$ to exert effort, for each $i \in S$ it must be that $\contract_i \ge {c_i}/f_S(i)$, where $f_S(i) = f(S) - f(S\setminus \{i\})$ denotes the marginal contribution of $i\in S$ to $S \setminus \{i\}$.
Thus, as also observed by \cite{duetting2022multi}, the optimal contract that incentivizes a set of agents $S$ to exert effort (i.e., the contract that does so with minimum expected payment), is given by $\contract_i = {c_i}/{f_S(i)}$ for $i\in S$ and $\contract_i=0$ otherwise.
{For convenience,} we interpret ${c_i}/{f_S(i)}$ as $0$ if $c_i = 0$ and $f_S(i)=0$, and as $\infty$ when $c_i > 0$ and $f_S(i)=0$. We denote by $p:2^\agents\rightarrow \reals_{\ge 0}$ the minimum total payment which incentivizes the set of agents $S$ to exert effort, i.e.,
\[
p(S) = \sum_{i\in S} \frac{c_i}{f_S(i)},
\quad
\text{where}
\quad
f_S(i) = f(S) - f(S\setminus \{i\}).
\]
The principal's expected utility from incentivizing a set of agents $S$ is given by
\begin{align*}
    g(S)=(1-p(S)) \cdot f(S).
\end{align*}
Thus, the problem of maximizing the principal's utility essentially reduces to finding a set of agents $S$ that maximizes $g(S)$.

\paragraph{Classes of Reward Functions.} We focus on reward functions $f$ that belong to one of the following classes of complement-free set functions \citep{lehmann2001combinatorial}. 
A set function $f:2^\agents\rightarrow \reals_{\ge 0}$ is:
\begin{itemize}
    \item \textit{additive} if there exist real non-negative values $\{v_i\}_{i \in \agents}$ such that $f(S) = \sum_{i\in S} v_i$ for all $S\subseteq \agents$.
    \item \textit{submodular} if for any two sets $S \subseteq S'\subseteq A$ and any $i\in S\subseteq S'$ it holds that $f_{S}(i) \ge f_{S'}(i)$.
    \item \textit{XOS} (also known as \emph{fractionally subadditive}) if there exists a finite collection of additive functions 
    $a_1, \ldots, a_k : 2^{\agents} \to \reals_{\ge 0}$
    such that for every $S\subseteq A$, it holds that $f(S) = \max_{i=1,\dots, k} a_i(S)$,
    \item \textit{subadditive} if for any two sets $S, S'\subseteq A$ it holds that $f(S\cup S') \leq f(S) + f(S')$.
\end{itemize}
It is well-known that
$$
\text{additive} 
\subsetneq
\text{submodular}
\subsetneq
\text{XOS}
\subsetneq
\text{subadditive}.
$$

\paragraph{Primitives for Accessing Set Functions.} 
{The reward function $f$ may have an exponentially large representation. A common way to address this challenge is by assuming oracle access to $f$. We consider the following standard primitives for querying the set function $f$:}
\begin{itemize}
    \item A \emph{value oracle} is given a set $S\subseteq \agents$ and returns $f(S)$.
    \item A \emph{demand oracle} is given a price vector $q\in \reals_{\ge 0}^\agents$ and returns a set $S\subseteq \agents $ that maximizes $f(S)-\sum_{i\in S} q_i$.
\end{itemize}
It is well known that demand oracles are stronger than value oracles in the sense that a value oracle can be simulated with poly-many calls to a demand oracle, but not vice-versa {\cite{nisan2007algorithmic}}.

\paragraph{Budget Constraints.} 
{In this paper, we consider settings where the principal is subject to a budget constraint on the contracts she can offer.}
A budget constraint is given by 
$B\in (0,1]$, which limits the total payment the principal can make to the agents. 
{Specifically, the principal may only incentivize a set of agents $S \subseteq \agents$ if 
$p(S) \leq B$. We refer to such sets $S$ as \emph{budget-feasible}.

\paragraph{Objectives and Maximization Problems.} Typically, {the} contract design literature {focuses on} the problem of finding a contract {that} maximizes the principal's utility, $g(S)$, {which we refer to as} \emph{profit}. 
In this paper, we go beyond this objective and also consider contracts that maximize additional objectives, including social welfare (defined as $f(S) - \sum_{i\in S} c_i$) {and} expected reward ({i.e., }$f(S)$). 

Any non-trivial objective depends on the instance specification (i.e., $f$ and $c$), and therefore cannot be treated as a fixed set function.
Since later sections involve proving formal reductions between instances, it is useful to think of an objective as a class of functions, with one set function per
instance.
Like $f$, these set functions are exponential in size. To avoid representation issues, we define objectives as follows.

\begin{definition}[Objectives]\label{def:objective}
    An \emph{objective} $\varphi$ is defined by a poly-time algorithm that is given a problem instance $\instance$ and a subset of agents $S\subseteq\agents$ and outputs a non-negative real number, denoted $\varphi_{\instance}(S)$. This algorithm is given value oracle access to $f$. 
\end{definition}
When the instance defining the objective is clear from context we omit it, 
and write simply $\varphi(S)$. 

\begin{definition}[{Maximization Problems}]\label{def:maxProb}
    For any given objective $\varphi$ and budget $B\in (0,1]$, the problem of $\MAX\varphi(B)$ is the computational problem of finding, given a problem instance $\instance$, a budget-feasible contract that maximizes $\varphi$. We also use $\MAX\varphi(B)$ to denote the optimal value of this problem given an instance $\instance$, i.e., $\MAX\varphi(B) = \max_{S\subseteq A\,:\, p(S) \le B} \varphi_{\instance}(S)$. We pay special attention to the following objectives, which we also give dedicated notation:
    \begin{enumerate}[label=(\roman*)]
        \item Expected reward: $\maxreward(B) = \max_{S \subseteq A \,:\, p(S) \leq B} f(S)$,
        \item Profit: $\maxrevenue(B) = \max_{S \subseteq A \,:\, p(S) \leq B} g(S)$ where $g(S) = (1-p(S)) \cdot f(S)$,
        \item Social welfare: $\maxwelfare(B) = \max_{S \subseteq A \,:\, p(S) \leq B} f(S) - c(S)$.
    \end{enumerate}
\end{definition}

Let $B \in (0,1]$ be a budget, and $\varphi$ be an objective.
We say that $S^\star$ is a \emph{solution to} $\MAX\varphi(B)$ if $p(S^\star)\le B$ and $\varphi(S^\star)=\MAX\varphi(B)$. 
Additionally, for $\gamma >1$ we say that $S$ is a $\gamma$-approximation to $\MAX\varphi(B)$ if $p(S)\le B$ and $\gamma \cdot \varphi(S) \ge \MAX\varphi(B)$.

\section{Structural Insights}\label{sec:structural_insights}
In this section, we present key structural insights.

In \Cref{sec:downsizing}, we present our \emph{downsizing lemmas}, which detail polytime algorithms that, given a team of agents $S\subseteq \agents$, return a subset $S'\subseteq S$ with reduced payment while maintaining a sizable fraction of the expected reward.
In \Cref{sec:light_agents} we define the set of \emph{light agents}, and the related problem of $\maxlightreward(B)$. Both of these notions are central to the analysis leading to our computational results. In particular, in \Cref{sec:computational_results} we show that $\maxlightreward(B)$ is equivalent to many interesting contract design problems.
In \Cref{sec:good-objectives}, we define and prove some key properties of what we call \goodobj\ {(BEyond STandard)} objectives, which are the objectives we study in our computational results. In particular, we show that profit, social welfare, expected reward, and any convex combination of the three are all \goodobj\ objectives.
In \Cref{sec:mulimulti_separation}, we observe that in the budgeted setting, even with a single agent, one cannot find a contract in which \emph{every} equilibrium yields a constant approximation to the optimal profit. This is in contrast to the results of \cite{multimulti}, which show that this is the case whenever $f$ is submodular\footnote{In fact, \cite{multimulti} show this result in a generalized setting where each agent may take any combination of actions. 
}.

\subsection{Downsizing Algorithms} \label{sec:downsizing}
In this section, we present our downsizing lemmas.
For any given team $S \subseteq \agents$, the downsizing lemmas allow us to remove a sufficiently large subset of agents from $S$ to meet a target budget constraint, while also giving a guarantee on the resulting reward for the set of the remaining agents. 

\begin{remark}
The scaling property of XOS functions \cite[Lemma~3.5]{duetting2022multi} (hereafter, the scaling lemma) provides a method for selecting a subset with a lower expected reward while maintaining sufficiently high marginal values.
Readers familiar with the scaling lemma might notice its connection to our downsizing lemmas. However, the scaling lemma does not guarantee a set with a lower total payment, which is the key property of our downsizing lemmas. Notably, to extend our downsizing lemma from submodular to XOS functions, we build on the approach of \cite{duetting2022multi} to establish \Cref{lem:xos_scaling_inter}, which, while similar to the scaling lemma, is incomparable to it.
\end{remark}
 
We first present the {downsizing} lemma for {submodular $f$.} In \Cref{sec:pof}, we demonstrate its tightness; see \Cref{rem:downsizing_opt}.

\begin{lemma}[Downsizing Lemma for Submodular Reward]\label{lem:submodular_payment_scaling}
    Let $\instance$ be an
    instance with submodular $f$ and let $\psi : 2^A \to [0,1]$ be any subadditive function.
    For any integer $M \geq 3$ and any subset of agents $S \subseteq \agents$, there exists a subset $T \subseteq S$ such that:
     \begin{align*}
      \left( p(T) \leq \frac{2}{M} \cdot p(S)  \quad \text{ or } \quad |T| = 1 \right) \quad \text{ and } \quad \psi(T) \geq \frac{1}{M-1} \cdot \psi(S).
    \end{align*} 
    Moreover, such set $T$ can be computed in polynomial time with value query access to $f$.
\end{lemma}

\begin{proof}
    We show that the output of  \Cref{alg:budget_scaling} satisfies the conditions of the lemma. 

\begin{algorithm}[t]
\caption{Downsizing Algorithm for Submodular Reward}\label{alg:budget_scaling}
\SetAlgoLined
\KwIn{integer $M \geq 3$ and a set $S \subseteq \agents$}
\KwOut{$T \subseteq S$ with $\psi(T) \geq \psi(S)/(M-1)$ and either $p(T) \leq (2/M) \cdot p(S)$ or $|T|=1$}
let $Z \gets \{i \in S \mid c_i/f_S(i) > (1/M) \cdot p(S)\}$\; 
\If{$\psi(\{i\}) \ge (1/(M-1)) \cdot \psi(S)$ for some $i \in Z$}{\label{line:singleton_if}\Return{$\{i\}$}\;\label{line:return_singleton}
}
let $U \gets S \setminus Z$\;
\For{$r = 1, \ldots, M-|Z|-2$}{\label{line:for_r} 
    set $W_r \gets \emptyset$\;
    \While{$U$ is non-empty \textbf{and} $\sum_{j \in W_r} c_j/f_S(j) \le (1/M) \cdot p(S)$}{\label{line:while_u}
        choose any agent $i \in U$\;
        $U \gets U \setminus \{i\}$\;
        $W_r \gets W_r \cup \{i\}$\;\label{line:add_agent}
    }
    \If{$\psi(W_r) \ge (1/(M-1)) \cdot \psi(S)$}{
        \label{line:second_if}     \Return{$W_r$}\;\label{line:return_set}
    }
}
\Return{$U$}\;\label{line:return_remainder}
\end{algorithm}
    First, it is clear that if the algorithm returns a singleton $\{i\}$ in \Cref{line:return_singleton}, the conditions of the lemma are met, by the if condition preceding it.
    Second, note that if the algorithm returns a set $W_r$ in \Cref{line:return_set}, then by the if condition, we have $\psi(W_r) \geq (1/(M-1)) \cdot \psi(S)$. Let $i$ be the last agent added to $W_r$ in \Cref{line:add_agent}. By the while-loop condition, we have $\sum_{j \in W_r \setminus \{i\}} c_j / f_S(j) \leq (1/M) \cdot \sum_{j \in S} c_j / f_S(j)$. Additionally, since $i \notin Z$, it follows that $c_i / f_S(i) \leq (1/M) \cdot \sum_{j \in S} c_j / f_S(j)$.    
    We conclude that
    \begin{align*}
        \sum_{j \in W_r} c_j/f_W(j) &\leq \sum_{j \in W_r} c_j/f_S(j) && (\text{by submodularity of $f$}) \\
        &= \sum_{j \in W_r \setminus \{i\}} c_j/f_S(j) + c_i/f_S(i)  \\
        &\leq (1/M) \cdot \sum_{j \in S} c_j/f_S(j) + (1/M) \cdot \sum_{j \in S} c_j/f_S(j) && (\text{by the above})\\
        &= (2/M) \cdot \sum_{j \in S} c_j/f_S(j)
    \end{align*}
Thus, if the algorithm returns a set in \Cref{line:return_set}, the conditions of the lemma are satisfied. 
Suppose that the algorithm returns the remaining agents $U$ in \Cref{line:return_remainder}. We have:
\begin{align*}
    \psi(U) &\geq \psi(S) - \sum_{i \in Z} \psi(\{i\}) - \sum_{r=1}^{M-|Z|-2} \psi(W_r) && (\text{by subadditivity of $\psi$})\\
    &\geq \psi(S) - (M-2) \cdot (1/(M-1)) \cdot \psi(S) && (\text{by \Cref{line:singleton_if} and \Cref{line:second_if}}) \\
    &= (1/(M-1)) \cdot \psi(S) 
\end{align*}
Since each element added to $W_1, \ldots, W_{M-|Z|-2}$ comes from $U$ and is simultaneously removed from $U$, these sets are pairwise disjoint. Thus, by submodularity of $f$, we have:
\begin{align*}
    \sum_{j \in U} c_j/f_U(j) &\leq \sum_{j \in U} c_j/f_S(j) \\
    &= \sum_{j \in S} c_j/f_S(j) - \sum_{i \in Z} c_i/f_S(i) - \sum_{r=1}^{M-|Z|-2} \sum_{i \in W_r} c_i/f_S(i) \\
    &\leq  \sum_{j \in S} c_j/f_S(j) - (|Z| + (M-|Z|-2)) \cdot (1/M) \cdot \sum_{j \in S} c_j/f_S(j) \\
    &= (2/M) \cdot \sum_{j \in S} c_j/f_S(j),
\end{align*}
where the second inequality follows by the definition of $Z$ and the while-loop condition.
This means that both of the conditions of the lemma are satisfied if the algorithm executes \Cref{line:return_remainder}, which concludes the proof.
\end{proof}

For XOS rewards, we first need to prove the following  property.

\begin{lemma}[Recovering Marginals of XOS Functions]\label{lem:xos_scaling_inter}
     Let $\instance$ be any 
     instance with XOS $f$. 
     For any sets $T \subseteq S \subseteq \agents$, there exists a subset $U \subseteq \agents$ such that:
     \begin{align*}
         f(U) \geq (1/2) \cdot f(T) \quad \text{ and } \quad f_U(i) \geq (1/2) \cdot f_S(i) \text{ for all $i \in U$}
     \end{align*}
     Moreover, such set $U$ can be computed in polynomial time with value query access to $f$.
\end{lemma}
\begin{proof}
    Let $U$ be the output of \Cref{alg:marginal_scaling} when given set $T$ as input. We will argue that $U$ satisfies the conditions of the lemma.

\begin{algorithm}[t]
\caption{Downsizing Algorithm for XOS Reward}\label{alg:marginal_scaling}
\SetAlgoLined
\KwIn{sets $T \subseteq S \subseteq \agents$}
\KwOut{set $U \subseteq T$ with $f(U) \geq (1/2) \cdot f(T)$ and $f_U(i) \geq (1/2) \cdot f_S(i)$ for all $i \in U$}
$U \gets T$\;
\While{there exists $i \in U$ such that $f_U(i) < (1/2) \cdot f_S(i)$}{
    $i^\star \gets \arg\min_{i \in U} f_U(i) / f_S(i) $\;
    $U \gets U \setminus \{i^\star\}$\;
}
\Return{$U$}\;
\end{algorithm}

The algorithm terminates because $|U|$ decreases in each iteration of the while loop. From the termination condition, it follows that $f_U(i) \geq (1/2) \cdot f_S(i)$ for all $i \in U$.

Next, we show that $f(U) \geq (1/2) \cdot f(T)$. 
Let $U_0 = T$ and define the sequence $U_1, \ldots, U_k$ as the sets $U$ throughout the execution, with removals $i_j = U_{j-1} \setminus U_j$ for $j = 1, \ldots, k$. By the while-loop condition and the choice of $i_j$, we know $f_{U_{j-1}}(i_j) < (1/2) \cdot f_S(i_j)$ for all $j = 1, \ldots, k$.

{We use the observation from \cite[Lemma 2.1]{duetting2022multi} that for any XOS function $f : 2^{\agents} \to [0,1]$ and sets $S_1 \subseteq S_2 \subseteq A$, it holds that $\sum_{i \in S_1} f_{S_2}(i) \leq f(S_1)$.}

Now, observe the following derivation:
\begin{align*}
    f(U) &= f(T) - \sum_{j=1}^{k} ( f(U_{j-1}) - f(U_j) ) && (\text{by the telescoping sum}) \\
         &= f(T) - \sum_{j=1}^{k} f_{U_{j-1}}(i_j) && (\text{by the definition of $f_{U_{j-1}}(i)$}) \\
         &> f(T) -  \frac{1}{2} \cdot \sum_{j=1}^{k} f_S(i_j) && (\text{by the observation above}) \\
         &= f(T) - \frac{1}{2} \cdot f(\{i_1, \ldots, i_k\}) && (\text{by \cite[Lemma 2.1]{duetting2022multi}}) \\
         &\geq \frac{1}{2} \cdot  f(T), && (\text{by monotonicity})
\end{align*}
which completes the proof.
\end{proof}

We now present the downsizing lemma for XOS $f$.
Notably, the guarantees provided for XOS $f$ are weaker than those for submodular $f$.

\begin{lemma}[Downsizing Lemma for XOS Reward]\label{lem:scale_xos}
    Let $\instance$ be any 
    instance with XOS $f$.
    For any integer $M \geq 3$ and any subset of agents $S \subseteq \agents$, there exists a subset $U \subseteq S$ such that:
     \begin{align*}
      \left( p(U) \leq \frac{4}{M} \cdot p(S)  \quad \text{ or } \quad |U| = 1 \right) \quad \text{ and } \quad f(U) \geq \frac{1}{2M-2} \cdot f(S).
    \end{align*} 
    Moreover, such set $U$ can be computed in polynomial time with value query access to $f$.
\end{lemma}
\begin{proof}
Applying \Cref{lem:submodular_payment_scaling} with $\psi = f$, we obtain a subset $T \subseteq S$ such that $f(T) \geq (1/(M-1)) \cdot f(S)$ and $\sum_{i \in T} c_i / f_S(i) \leq (2/M) \cdot \sum_{i \in S} c_i / f_S(i).$  Applying \Cref{lem:xos_scaling_inter} with $T \subseteq S \subseteq \agents$, we obtain a subset $U \subseteq T$ such that $f(U) \geq (1/2) \cdot f(T) \geq (1/2) \cdot (1/(M-1)) \cdot f(S)$ and:
\[
    p(U) = \sum_{i \in U} \frac{c_i}{f_U(i)} \leq 2 \cdot \sum_{i \in U} \frac{c_i}{f_S(i)} \leq 2 \cdot \sum_{i \in T} \frac{c_i}{f_S(i)}  \leq \frac{4}{M} \cdot \sum_{i \in S} \frac{c_i}{f_S(i)} = \frac{4}{M} \cdot p(S).
\]
which completes the proof.
\end{proof}

\subsection{Light Agents} \label{sec:light_agents}
We now present the notion of \emph{light agents}, which is key to our analysis. 
\begin{definition} [The Set of Light Agents]
   For any instance $\instance$, the set of light agents, denoted by $\lightagents$,  is the set of agents $i$ such that $p(\{i\}) \le 1/2$, i.e.,
    \[
    \lightagents = \left\{i\in A\mid {c_i}/{f(\{i\})} \le {1}/{2}\right\}.
    \]
    {We refer to the set $\agents \setminus L$ as the set of heavy agents.}
\end{definition}
Note that the set of light agents was also central to the analysis of \cite{duetting2022multi}.
Light agents are important because each can be incentivized to exert effort individually with a payment of at most $1/2$, ensuring a profit that is a constant fraction of the reward.
Another key property of the light agents set is that, for any subadditive $f$, a budget-feasible contract with a budget of at most $1$ can incentivize at most one heavy agent to exert effort. 
{This is cast in the following observation, which has essentially been proved in \cite[Lemma 3.2]{duetting2022multi}.}

\begin{observation}[Budget-Feasible Sets Contain at Most One Heavy Agent]\label{obs:1_heavy}
    Fix an instance $\instance$ with subadditive $f$ and a budget $B\le 1$. If $S\subseteq \agents$ is budget-feasible, then $|S\setminus \lightagents| \le 1$.
\end{observation}
\begin{proof}
    Assume towards contradiction that $S\subseteq \agents$ is budget-feasible, and that $|S\setminus\lightagents|>1$. For any heavy agent $i\in S\setminus\lightagents$ it holds that 
    \[
    \frac{c_i}{f_S(i)} = \frac{c_i}{f(S) - f(S \setminus \{i\})} \ge \frac{c_i}{f(\{i\})} >\frac{1}{2},
    \]
    where the first inequality is from subadditivity of $f$. Thus, 
    \[
    p(S) =\sum_{i\in S}\frac{c_i}{f_S(i)} \ge \sum_{i\in S\setminus \lightagents}\frac{c_i}{f_S(i)} > |S\setminus \lightagents| \cdot (1/2) \ge 1 \ge B,
    \]
    contradicting the budget-feasibility of $S$.
\end{proof}
A key aspect of our analysis of the optimal contract problem is maximizing reward among light agents while ensuring budget feasibility. We formally define this problem as follows.
\begin{definition}[$\maxlightreward(B)$]\label{def:mlr}
    Let $\instance$ be an instance.
    For any $B\in (0,1]$, the problem of $\maxlightreward(B)$ is defined as 
    \[
    \maxlightreward(B) = \max_{S\subseteq \lightagents: p(S)\le B} f(S).
    \]
\end{definition}
In \Cref{sec:computational_results}, we show that this problem is equivalent up to a constant to several important contract design problems.

\subsection{BEST Objectives}\label{sec:good-objectives}
One of our key insights is that many important optimization problems in multi-agent contract design are essentially equivalent. We establish this for a broad class of objectives, which we call BEyond STandard (BEST). As we show below, this class encompasses standard objectives such as reward, social welfare (often termed the first best in economics), and profit (commonly referred to as the second best).
We formally define this class as follows.

\begin{definition}[{Beyond Standard (BEST) Objectives}]  \label{def:goodobj} 
    An objective $\varphi$ {belongs to the class of  beyond standard (BEST) objectives} if, for any instance $\instance$ and any $S\subseteq \agents$, it holds that:
    \begin{enumerate}[label=(\roman*)]
        \item $\varphi$ 
        is sandwiched between the profit and the reward, i.e., $g(S) \le \varphi(S) \le f(S)$,
        \item for any agent $i\in S$ it holds that $\varphi(S) \le f(S\setminus \{i\})+\varphi(\{i\})$.
    \end{enumerate}
\end{definition}

We first show that when $f$ is subadditive, then both 
the reward and social welfare belong to the \goodobj\ class of objectives.

\begin{lemma}
    Any subadditive objective $\varphi$ that is
    sandwiched between the profit and reward is \goodobj.
\end{lemma}
\begin{proof}
   The second condition of \goodobj\ objectives follows from the inequality  
 $\varphi(S) \leq \varphi(S \setminus\{i\}) + \varphi(\{i\}) \leq f(S \setminus\{i\}) + \varphi(\{i\})$, where the first inequality holds by the subadditivity of $\varphi$, and the second follows from the fact that $\varphi$ is dominated by $f$.
\end{proof}

We now show that the profit also belongs to the class.

\begin{lemma}
    When $f$ is subadditive, $g$ belongs to the \goodobj\ class of objectives.
\end{lemma}
\begin{proof}
   Trivially, $g$ satisfies condition (i) of Definition~\ref{def:goodobj}. To verify condition (ii), let $S \subseteq \agents$ and $i \in S$. 
   First, observe that:
   \begin{align*}
   p(S) = \sum_{j \in S} c_j/f_S(j) \geq c_i/f_S(i) = c_i/(f(S) - f(S \setminus \{i\})) \geq c_i/f(\{i\}) = p(\{i\}).    
   \end{align*}
   by subadditivity of $f$.
   Then, we can bound $g(S)$ as follows:
   \begin{align*}
       g(S) = (1-p(S)) \cdot f(S) \leq (1-p(\{i\})) \cdot f(S) \leq f(S \setminus \{i\}) + (1-p(\{i\})) \cdot f(\{i\})
   \end{align*}
    as needed.
\end{proof}

Additionally, we note that the class of \goodobj \ objectives is closed under convex combinations.
\begin{observation}
    Let $\varphi^1,\dots,\varphi^k$ be \goodobj\ objectives, and let $\lambda_1,\dots,\lambda_k\in (0,1)$ be such that $\sum_{i=1}^k\lambda_i =1$. It holds that the objective $\varphi$ defined as $\varphi_{\instance}(S) = \sum_{i=1}^k\lambda_i \varphi^i_{\instance}(S)$ is \goodobj.
\end{observation}
{\begin{proof}
    Let $\instance$ be and instance and let $S\subseteq \agents$.
    We start by showing condition (i) of \Cref{def:goodobj}. It holds that 
    \[
    \varphi(S) = \sum_{i=1}^k \lambda_i \varphi_i(S) \ge \sum_{i=1}^k \lambda_i g(S) = g(S) \quad\text{and}\quad
    \varphi(S) = \sum_{i=1}^k \lambda_i \varphi_i(S) \le \sum_{i=1}^k \lambda_i f(S) = f(S),
    \]
    thus $g(S)\le \varphi(S)\le f(S)$, as needed.
    
    We now turn our attention to condition (ii) of \Cref{def:goodobj}. Let $i\in S$. We have
    \[
    \varphi(S) = \sum_{j=1}^k \lambda_j \varphi_j(S) \le \sum_{j=1}^k \lambda_j\left(f(S\setminus \{i\})+\varphi_j(\{i\})\right) =f(S\setminus \{i\}) + \sum_{j=1}^k \lambda_j \varphi_j(\{i\})=f(S\setminus \{i\})+\varphi(\{i\}),
    \]
    as needed.
\end{proof}
}
The following {key} property of \goodobj\ objectives {motivates} our focus on them.
\begin{lemma}  [{Key Property of \goodobj\ Objectives}] \label{lemma:goodobj_upper_bound}
 Fix an instance $\instance$ with XOS $f$, a budget $B\in (0,1]$, and a \goodobj  \
    objective $\varphi$.
    It holds that
\[
\MAX\varphi(B) \le 2 \cdot \maxlightreward(B)+\max_{i\in \agents} \varphi(\{i\}).
\]
\end{lemma}
\begin{proof}
    Let $S^\star \subseteq \agents$ be an optimal solution to $\MAX\varphi(B)$, i.e., $\varphi(S^\star) = \MAX\varphi(B)$ and $p(S^\star) \le B$.
    Note that, by \Cref{obs:1_heavy}, $S^\star$ contains at most one heavy agent. If it contains no heavy agents, we have 
    $\varphi(S^\star) \le f(S^\star) = \maxlightreward(B)$, as needed.
    Otherwise, let $i^\star$ be such that $S^\star \setminus \lightagents= \{i^\star\}$ and we get:
    \begin{align*}
    \varphi(S^\star)
    &\le
    \varphi(S^\star \cap \lightagents) + \varphi(\{i^\star\}) && (\text{by property $1$ of \goodobj\ objectives}) \\
    &\le
    f(S^\star \cap \lightagents) + \max_{i \in A} \varphi(\{i\}) && (\text{since $\varphi$ is dominated by $f$})  
    \end{align*}
    It remains to show that $f(S^\star \cap \lightagents) \le 2 \cdot \maxlightreward(B)$. 
    To see this, we apply  \Cref{lem:xos_scaling_inter} with $T=S^\star \cap \lightagents$ and $S=S^\star$. This yields a set  $U \subseteq S^\star \cap \lightagents$ such that $f(U) \ge (1/2) \cdot f(S^\star \cap \lightagents)$ and $f_U(i) \ge (1/2) \cdot {f_{S^\star}(i)}$ for every $i \in U$. In particular, we get that $U$ is budget-feasible, since
    \begin{align*}
        p(U) 
        &= 
        \sum_{i \in U} \frac{c_i}{f_U(i)} \\
        &\le 
        \sum_{i \in S^\star \cap \lightagents} \frac{2 \cdot c_i}{f_{S^\star}(i)} && (\text{since }f_U(i) \ge (1/2) \cdot {f_{S^\star}(i)}) \\
        &=
        2 \cdot \left(p(S^\star) - \frac{c_i}{f_{S^\star}(i^\star)}\right) && (\text{since } S^\star \setminus \lightagents = \{i^\star\})\\
        &\le
        2 \cdot \left(p(S^\star) - \frac{c_i}{f(i^\star)}\right) && (\text{by subadditivity of } f)\\
        & \le
        2(B-1/2)    && (\text{since } S^\star \text{ is budget feasible and } i^\star \notin \lightagents) \\
        &\le B && (\text{since } B\le 1).
    \end{align*}
    Thus $f(S^\star \cap \lightagents) \le 2 \cdot  f(U) \le 2\cdot \maxlightreward(B)$, concluding the proof.
\end{proof}

\subsection{Bad Equilibria Are Inevitable}\label{sec:mulimulti_separation}
In this section, we show a separation between the budgeted and unbudgeted settings. Specifically, it was shown in \cite{multimulti} that one can find a contract $\contract$ such that \emph{every} equilibrium $S \in \nash(\contract)$ gives a constant-factor approximation to the optimal profit\footnote{The work of \cite{multimulti} considers the more general case in which any agent can take any subset of a given set of actions. However, the separation holds even for the case of binary actions {(as considered in \cite{duetting2022multi} and here)}. 
}. We show that such a guarantee is not possible in the budgeted setting, even for a single agent, 
as there are instances in which every contract has an equilibrium which does not approximate the optimal profit.
This is in contrast with the unbudgeted setting, for which it is known that only when $f$ is XOS (a strict super class of submodular functions), there are instances such that every contract has a ``bad'' equilibrium (see Proposition B.2 in \cite{multimulti}).

\begin{proposition}
    There exists a single-agent {binary-actions} instance 
    $\langle A, f, c \rangle$ {with $\agents = \{i\}$}
    and a budget $B=1/2$ such that:
    \begin{enumerate}[label=(\roman*)]
        \item There exists a budget feasible contract $\contract_i$ and an equilibrium $S \in \nash(\contract_i)$ with $g(S) = 1/2$. 
        \item For any budget feasible contract $\contract_i \le B$, there exists an equilibrium $S \in \nash(\contract_i)$ with $g(S) = 0$.
    \end{enumerate}
\end{proposition}
\begin{proof}
    Consider a single agent, $A = \{i\}$, whose cost for exerting effort is $c = {1}/{2}$ and the success probability if he exerts effort is $f(\{i\})=1$.

    It is easy to verify that any contract $\contract_i < 1/{2}$ has a unique equilibrium $\nash(\contract_i)=\emptyset$, in which the principal's profit is 0.
    For $\contract_i={1}/{2}$, we have two equilibria: the one in which the agent exerts effort, yielding a profit of ${1}/{2}$, and an expected utility of 0 to the agent:
    $\contract_i \cdot f(\{i\}) - c_i = {1}/{2} - {1}/{2} = 0$.
    Thus, shirking is also an equilibrium for $\contract_i={1}/{2}$ (i.e., $\emptyset \in \nash({1}/{2})$), for which the optimal profit is ${1}/{2}$.
    The claim follows as any other contract is not budget feasible.
\end{proof}

\section{All Objectives Are Equivalent} \label{sec:computational_results}

In this section we show that all \goodobj\ objectives {(see \Cref{def:goodobj})} and all budget {constraints} are equivalent 
up to a constant {factor}.
More formally, we show the following:
\begin{theorem} [Equivalence of All \goodobj\ Objectives and Budgets] \label{thm:objective_equivalnce}
    Fix any two \goodobj\ objectives $\varphi, \varphi'$ and any two budget $B,B'\in (0,1]$.
    For XOS $f$, there exists a poly-time reduction  from $\MAX\varphi(B)$ to $\MAX\varphi'(B')$ that loses only a constant {factor} in the approximation. {Moreover,} this reduction requires only value oracle access to $f$.
\end{theorem}

The implications of this theorem are particularly strong together with a main result of \citep{duetting2022multi}; namely, a poly-time $O(1)$-approximation to the $\maxrevenue(1)$ problem\footnote{They show a poly-time $O(1)$-approximation to the best contract with respect to profit with no budget constraints, but we note that maximizing profit induces an implicit budget of $1$, since otherwise the profit is negative.}:

\begin{theorem} [\citep{duetting2022multi}] \label{thm:duetting_approx}
{When $f$ is XOS,}
there exists a poly-time algorithm which achieves an $O(1)$-approximation to $\maxrevenue(1)$ with demand oracle access to $f$. Moreover, if $f$ is submodular, the same approximation can be obtained using only a value oracle.
\end{theorem}

Our reductions in \Cref{thm:objective_equivalnce} together with 
\Cref{thm:duetting_approx}
implies the following:
\begin{corollary}[{Constant-Factor Approximations Under Budget Constraints}]\label{Cor:constant_approx}
   {When $f$ is XOS,} for any {objective} $\varphi$ in the class of \goodobj\ objectives (including profit, reward, and welfare), there exists a polynomial-time $O(1)$-approximation algorithm to $\MAX\varphi(B)$ using a demand oracle. Moreover, if $f$ is submodular, the same guarantee holds with only a value oracle.
\end{corollary}

{Notably, 
\cite{ezra2023Inapproximability} show \Cref{thm:duetting_approx} to be tight in the sense that 
demand queries cannot be replaced with value queries
in the case of XOS $f$. 
Our reduction in \Cref{thm:objective_equivalnce} extends this {impossibility} result
to any \goodobj\ objective and any budget.
}

\begin{remark}\label{remark:unconstrained_sw}
    {Our positive result in}
    \Cref{Cor:constant_approx} 
    {for welfare maximization}
    {in submodular instances}
    demonstrates a stark contrast to 
    {settings with no budget constraints.}
    {Specifically, if we set $B = \infty$, potentially allowing the principal's utility to be negative, finding}
    an approximately optimal contract with respect to social welfare is 
    {equivalent to approximately solving a demand query}
    with respect to {submodular} $f$ and prices $c$. It is well-known that unless $\mathsf{P}=\mathsf{NP}$, such a set cannot be approximated in polynomial time with value oracle access \cite{feige2014demand}. 
\end{remark}
The proof of \Cref{thm:objective_equivalnce} relies on {a reduction to} the core problem  $\maxlightreward(B)$. We {establish} \Cref{thm:objective_equivalnce} in two {steps}: 
\begin{enumerate}[label=(\arabic*)]
    \item We reduce from $\MAX\varphi_1(B)$ to $\maxlightreward(B)$ for any \goodobj\ objective $\varphi_1$ and budget $B\in (0,1]$ (\Cref{lem:reduction_to_mrl}).
    \item We reduce from $\maxlightreward(B)$ to $\MAX\varphi_2(B')$ 
 for any \goodobj\ objective $\varphi_2$ and budgets $B,B'\in (0,1]$ (\Cref{lem:reduction_from_mrl}). 
\end{enumerate} 
Both reductions lose only a constant {factor} in the approximation and run in poly-time with value oracle access.

\begin{remark}
    {We state our reductions below (\Cref{lem:reduction_to_mrl,lem:reduction_from_mrl}) for XOS $f$. In \Cref{app:objective_reductions_submodular}, we provide versions with improved constants under the stronger assumption that $f$ is submodular.}
\end{remark}

First, we prove the following reduction from $\MAX\varphi(B)$ to $\maxlightreward(B)$ .

\begin{lemma}[Reduction to $\maxlightreward(B)$]
\label{lem:reduction_to_mrl}
    Fix an instance $\instance$ with XOS $f$, a budget $B\in (0,1]$, and a \goodobj  \
    objective $\varphi$.
    For any given team $S \subseteq \agents$ that is a $\gamma$-approximation to $\maxlightreward(B)$, let $S'$ be the result of applying \Cref{lem:scale_xos} to $S$ with $M=5$. Then, it holds that one of $\{\{i\}\}_{i\in \agents} \cup \{S'\}$ is a $(40\gamma+1)$-approximation to $\MAX\varphi(B)$.
\end{lemma}

\begin{proof}
    By the guarantees of \Cref{lem:scale_xos} we have:
    \[f(S')\ge (1/4) \cdot f(S^\star) \ge (1/(4\gamma)) \cdot \maxlightreward(B).
    \]
    Moreover, we have either $|S'| = 1$ or $p(S') \leq (4/5) \cdot p(S)$. Observe that $S'$ is budget-feasible since, if $|S'|\ne 1$ then $p(S')\le (4/5) \cdot p(S)\le  (4/5) \cdot B\le B$, and if $|S'|=1$ then from subadditivity of $f$,
    \[
    p(\{i\})=\frac{c_i}{f(\{i\})}\le \frac{c_i}{f(S)-f(S\setminus\{i\})}=\frac{c_i}{f_S(i)}\le \sum_{j\in S}\frac{c_j}{f_S(j)}=p(S)\le B.
    \]
    
    We will now show that $p(S') \leq 4/5$. If $|S'| = 1$, then since $S' \subseteq \lightagents$, it follows that $p(S') \leq 1/2$.
    If $p(S') \leq (4/5) \cdot p(S)$, then by the budget feasibility of $S$, we get $p(S') \leq (4/5) \cdot p(S) \leq (4/5) \cdot B \le 4/5$. 

    Next, since the profit is dominated by $\varphi$, we observe that:
    \[
    \varphi(S') \ge g(S') = (1-p(S')) \cdot f(S') \ge \frac{1}{5} \cdot \frac{1}{4}  \cdot f(S) = \frac{1}{20  \gamma} \cdot \maxlightreward(B).
    \]
    Thus, by the key property of \goodobj\ objectives (\Cref{lemma:goodobj_upper_bound}), we obtain:
    \begin{align*}
    \MAX\varphi(B)&\le 2\cdot \maxlightreward(B)+\max_{i\in \agents} \varphi(\{i\}) \le (40\gamma) \cdot \varphi(S')+\max_{i\in \agents} \varphi(\{i\}),
    \end{align*}
    as needed.
\end{proof}

Next, we prove a reduction from $\maxlightreward(B)$ to $\MAX\varphi(B')$ {for any budget $B'$}.

\begin{lemma}
[Reduction from $\maxlightreward(B)$]
 \label{lem:reduction_from_mrl}
    Let $\mathcal{I} = \instance$  be an 
    instance 
    with XOS $f$, let $\varphi$ be a \goodobj\  objective, and let $B, B'\in(0,1]$.
    Consider the instance $\mathcal{I}'=\langle \lightagents, f\mid_\lightagents, c \cdot (B'/B)\mid_\lightagents\rangle$, which is the same as $\mathcal{I}$ except with only the light agents, and with costs scaled by ${B'}/{B}$. 
    {Then,} if $S$ is a $\gamma$-approximation to $\MAX\varphi(B')$ in {instance} $\mathcal{I}'$, then one of $\{\{i\}\}_{i\in \agents} \cup \{S\}$ is a $20\gamma$-approximation to $\maxlightreward(B)$ in {instance} $\mathcal{I}$.
\end{lemma}

\begin{proof} 
Let $S^\star$ be a solution to $\maxlightreward(B)$, i.e., $f(S^\star) = \maxlightreward(B)$ and we have $S^\star \subseteq \lightagents$ and $p(S^\star) \le B$. 
Apply \Cref{lem:scale_xos} to $S^\star$ with $M=5$, and get a set $T$ such that $f(T) \ge (1/4) \cdot f(S^\star)$ and either $|T|=1$ or $p(T) \le (4/5) \cdot  p(S) \le (4/5) \cdot B$. Observe that, in either case, $p(T)\le p(S^\star)\le B$.

Note that if $|T| = 1$, then $T$ itself is a $4$-approximation to $\maxlightreward(B)$ since $f(T) \geq (1/4) \cdot f(S^\star)$, $T \in L$, and $p(T) \leq B$. Therefore, it is also a $20\gamma$-approximation.

Otherwise, let $p'$ and $g'$ denote\footnote{When considering $p,g$ as objectives this is essentially the notation $p'=p_{\mathcal{I}'}$ and $g'=g_{\mathcal{I'}}$.} the total payment and the principal's profit (respectively) in the scaled instance $\mathcal{I}'$. Now, for each $S \subseteq L$, we have $p'(S) = (B'/B) \cdot \sum_{i \in S}  c_i/f_S(i)$ and $g'(S) = (1-p'(S)) \cdot f(S)$. Note that $p'(T) =({B'}/{B}) \cdot p(T) \le (4/5) \cdot B'\le 4/5$, and therefore
$T$ is budget-feasible in $\mathcal{I}'$ with respect to budget $B'$. We thus it holds that
\[
\begin{split}
f(S) &\ge \varphi_{\mathcal{I'}}(S) \ge (1/\gamma) \cdot \varphi_{\mathcal{I'}}(T) \ge (1/\gamma) \cdot g'(T)=(1/\gamma) \cdot (1-p'(T)) \cdot f(T) \\
&\ge (1/\gamma) \cdot (1/5) \cdot (1/4)  \cdot f(S^\star) = (1/(20\gamma)) \cdot  \maxlightreward(B).
\end{split}
\]
Additionally, $S$ is budget-feasible in the original instance $\mathcal{I}$, since $p(S) = (B/B') \cdot p'(S) \le (B/B') \cdot B' = B$, concluding the proof.
\end{proof}

\section{Price of Frugality}\label{sec:pof}
In this section, we analyze the price of frugality, which quantifies the loss in the principal's reward, profit, or social welfare due to budget constraints. 
We define the price of frugality for a given objective as the worst-case loss in that objective due to a budget reduction from $B$ to $b < B$.
In general, this loss can be unbounded, as the only bundle incentivizable within budget $b$ may be the empty set.
To address this, we impose the necessary assumption that all singleton bundles are incentivizable at $b$.

\begin{definition}[Price of Frugality]\label{def:pof}
    For a given 
    instance $\instance$ and an objective $\varphi$, the price of frugality for budget $b > 0$ with respect to budget $B > b$ is:
    \begin{align*}
        \gapphi(b,B) = {\MAX\varphi(B)}/{\MAX\varphi(b)}
    \end{align*}
    Furthermore, we define the worst-case price of frugality for a given class of instances $\mathcal{I}$
    under the additional constraint that singletons are budget-feasible, i.e., $\max_{i \in \agents} p(\{i\}) \leq b$, as follows:
    \begin{align*}        \gapphi\text{-}\mathcal{I}(b,B) = \sup_{\instance \in \mathcal{I}_b} \gapphi(b,B) \quad \text{where} \quad  \mathcal{I}_b = \{\instance \in \mathcal{I} : \max_{i \in \agents} p(\{i\}) \leq b\}.
    \end{align*} We define $\gapphiadditive(b,B)$, $\gapphisubmodular(b,B)$, and $\gapphixos(b,B)$ by considering the classes of all instances with additive, submodular, and XOS $f$, respectively.

     For the specific case where $\varphi = f$, we refer to $\gapphi$ as $\gapreward$. When $\varphi = g$, we call it $\gaprevenue$. Similarly, for $\varphi = f - c$, we refer to $\gapphi$ as $\gapwelfare$. We use the same naming conventions for the submodular and XOS price of frugality.
\end{definition}

Our main result is the following asymptotic bound on the price of frugality for XOS  $f$.

\begin{restatable}[Asymptotic Bounds on Price of Frugality]
    {theorem}{asymptoticfrugality} \label{thm:price_asymptotic}
     For any $0 < b < B \leq 1$ and any \goodobj\  objective $\varphi$, the price of frugality for XOS instances satisfies $\gapphixos(b,B) = O(\min(B/b, n))$. Moreover, this bound is tight even for additive instances, as $\gapphiadditive(b,B) = \Omega(\min(B/b, n))$. 
\end{restatable}
The proof of \Cref{thm:price_asymptotic} is deferred to \Cref{sec:price_of_f_proofs}.

The remainder of section is structured as follows: In \Cref{sec:pof_submod}, we derive the exact price of frugality for reward and welfare when 
$f$ is submodular and discuss its connections to downsizing algorithms.
Next, in \Cref{sec:pof_xossubadd}, we present two separation results: one for XOS $f$ and another for subadditive $f$.
Finally, in \Cref{sec:pof_revene}, we provide upper and lower bounds for the price of frugality for profit when $f$ is submodular, leaving the gap between them as an interesting open problem.

\subsection{Price of Frugality in Submodular Instances}\label{sec:pof_submod}

\begin{figure}[t]
    \centering
    \begin{subfigure}[b]{0.48\textwidth}
        \centering
\begin{tikzpicture}
    \begin{axis}[
        width=\myheight, height=\myheight,
        xlabel={budget $b$},
        xmin=0, xmax=1,
        ymin=0, ymax=1,
        grid=both,
        legend pos=south east,
        xtick={0, 2/2, 2/3, 2/4, 2/5,  2/10},
        xticklabels={$0$, $1$, $2/3$, $2/4$, $2/5$, $2/n$},
        ytick={1, 1/2, 1/2, 1/3, 1/4, 1/5, 1/10, 0},
        yticklabels={$1$, $1/2$, \phantom{$3/16$}, $1/3$, $1/4$, $1/5$, $1/n$, $0$},
        tick label style={font=\small},
    ]

    \addplot[black, only marks, mark=*, mark size=1pt] coordinates {(1,1)};
    \addplot[black, only marks, mark=o, mark size=1pt] coordinates {(1,0.5)};

    \foreach \M in {3,4,5,6,7,8,9,10} {
        \addplot[black, thick] coordinates 
            {({2/\M},{1/(\M-1)}) ({2/(\M-1)},{1/(\M-1)})};
            
        \addplot[black, only marks, mark=*, mark size=1pt] coordinates {({2/\M},{1/(\M-1)})};

        \addplot[black, only marks, mark=o, mark size=1pt] coordinates {({2/\M},{1/\M})};
    }
    
     \addplot[black, thick] coordinates 
            {({0},{1/10}) ({2/10},{1/10})};

    \addplot[dashed, thin, opacity=0.5] 
        coordinates {(0,0) (1,0.5)};

    \addplot[green!60!black, thick, domain=0.37228:1, samples=200] {1/(2-x)};
    \addplot[green!60!black, thick, domain=0.2857:0.37228, samples=100] {(2-x)/((2-3*x)*3)};
    \addplot[green!60!black, thick, domain=0.2222:0.2857, samples=100] {(2-x)/((2-4*x)*4)};
    \addplot[green!60!black, thick, domain=0.181818:0.2222, samples=100] {(2-x)/((2-5*x)*5)};
    \addplot[green!60!black, thick, domain=0.15384:0.181818, samples=100] {(2-x)/((2-6*x)*6)};
    \addplot[green!60!black, thick, domain=0.13333:0.15384, samples=100] {(2-x)/((2-7*x)*7)};
    \addplot[green!60!black, thick, domain=0.117647:0.13333, samples=100] {(2-x)/((2-8*x)*8)};
    \addplot[green!60!black, thick, domain=0.105263:0.117647, samples=100] {(2-x)/((2-9*x)*9)};
    \addplot[green!60!black, thick, domain=0:0.105263, samples=100] {(2-x)/((2-10*x)*10)};

    \end{axis}
\end{tikzpicture}

    \caption{
    The price of frugality bounds of \Cref{thm:main_price} for $B=1$ and $n=10$.
    The black line depicts 
    $1/\gaprewardsubmodular(b,B)$. The dashed line goes through the line $b/2$. The green line represents our upper bound on  $1/\gaprevenuesubmodular(b,B)$.}
    \label{fig:price_of_frugality}
    \end{subfigure}
    \hfill
    \begin{subfigure}[b]{0.48\textwidth}
        \centering
    \begin{tikzpicture}
    \begin{axis}[
        width=\myheight, height=\myheight,
        xlabel={budget $b$},
        ylabel={},
        xmin=0, xmax=1,
        ymin=0, ymax=1,
        legend pos=south east,
        xtick={0, 1/4, 0.4, 1/2, 3/4, 1},
        xticklabels={$0$, $1/4$, $2/5$, $1/2$, $3/4$, $1$},
        ytick={0, 3/16, 1/4, 2/4, 3/4, 1},
        yticklabels={$0$,  $3/16$, $1/4$, $1/2$, $3/4$, $1$},
        tick label style={font=\small},
    ]

    \addplot[blue, thick] coordinates {(0,0) (1/4,0)};
    \addplot[blue, only marks, mark=*, mark size=1pt] coordinates {(0,0)};
    \addplot[blue, only marks, mark=o, mark size=1pt] coordinates {(1/4,0)};

    \addplot[red, thick] coordinates {(0,0.01) (1/4,0.01)};
    \addplot[red, only marks, mark=*, mark size=1pt] coordinates {(0,0.01)};
    \addplot[red, only marks, mark=o, mark size=1pt] coordinates {(1/4,0.01)};

    \addplot[red, thick] coordinates {(1/4,1/4) (1/2,1/4)};
    \addplot[red, thick] coordinates {(1/2,2/4) (3/4,2/4)};
    \addplot[red, thick] coordinates {(3/4,3/4) (1,3/4)};
    \addplot[red, only marks, mark=*, mark size=1pt] coordinates {(1/4,1/4) (1/2,2/4) (3/4,3/4)};
    \addplot[red, only marks, mark=*, mark size=1pt] coordinates {(1,1)};
    \addplot[red, only marks, mark=o, mark size=1pt] coordinates {(1/2,1/4) (3/4,2/4) (1,3/4)};

    \addplot[blue, thick] coordinates {(1/4,3/16) (1/2,3/16)};
    \addplot[blue, thick] coordinates {(1/2,6/16) (3/4,6/16)};
    \addplot[blue, thick] coordinates {(3/4,9/16) (1,9/16)};
    \addplot[blue, only marks, mark=*, mark size=1pt] coordinates {(1/4,3/16) (1/2,6/16) (3/4,9/16)};
    \addplot[blue, only marks, mark=*, mark size=1pt] coordinates {(1,12/16)};
    \addplot[blue, only marks, mark=o, mark size=1pt] coordinates {(1/2,3/16) (3/4,6/16) (1,9/16)};

    \addplot[black] coordinates {(0.4,0) (0.4,1/4)};

    \addplot[fill=red, fill opacity=0.2, draw=none] 
        coordinates {(1/4,3/16) (1/2,3/16) (1/2,1/4) (1/4,1/4) (1/4,3/16)};

    \addplot[fill=red, fill opacity=0.2, draw=none] 
        coordinates {(1/2,6/16) (3/4,6/16) (3/4,1/2) (1/2,1/2) (1/2,6/16)};

    \addplot[fill=red, fill opacity=0.2, draw=none] 
        coordinates {(3/4,9/16) (1,9/16) (1,3/4) (3/4,3/4) (3/4,9/16)};

    \addplot[fill=blue, fill opacity=0.2, draw=none] 
        coordinates {(1/4,0) (1/2,0) (1/2,3/16) (1/4,3/16) (1/4,0)};

    \addplot[fill=blue, fill opacity=0.2, draw=none] 
        coordinates {(1/2,0) (3/4,0) (3/4,6/16) (1/2,6/16) (1/2,0)};

    \addplot[fill=blue, fill opacity=0.2, draw=none] 
        coordinates {(3/4,0) (1,0) (1,9/16) (3/4,9/16) (3/4,0)};

    \draw[decorate, decoration={brace, amplitude=5pt, mirror}, thick] 
        (axis cs:0.4,0) -- (axis cs:0.4,3/16) 
        node[midway, xshift=44pt] {\small $\maxwelfare(\frac{2}{5})$};

    \draw[decorate, decoration={brace, amplitude=5pt}, thick] 
        (axis cs:0.4,0) -- (axis cs:0.4,1/4) 
        node[midway, xshift=-42pt] {\small $\maxreward(\frac{2}{5})$};

    \draw[decorate, decoration={brace, amplitude=16pt, aspect=0.45}, thick] 
        (axis cs:1,0) -- (axis cs:1,1) 
        node[pos=0.53, anchor=west, yshift=-16pt, xshift=-95pt] {\small $\maxreward(B)$};

    \draw[decorate, decoration={brace, amplitude=6pt, aspect=0.33}, thick] 
        (axis cs:1,0) -- (axis cs:1,3/4) 
        node[pos=0.33, xshift=-50pt] {\small $\maxwelfare(B)$};

    \end{axis}
\end{tikzpicture}

    \caption{The instance from the proof of \Cref{lem:pof_upper} with $M=4$ and $B=1$. The red line is $\maxreward(b)$ and the blue line is $\maxwelfare(b)$. 
    For budget $b = 2/5$,   
        the price of frugality for both reward and welfare is $1/4$. \\}
    \label{fig:impo}
    \end{subfigure}
    \caption{
    Illustrating a connection between $\gaprewardsubmodular(b,B)$ and $\maxreward(b)$. 
    Observe that in the region where $b \geq 1/4$, the red lines in \Cref{fig:impo}, which represent the maximum reward, lie above the corresponding black lines in \Cref{fig:price_of_frugality}, which represent the price of frugality.
    In fact, this holds more generally. By definition of the price of frugality, for any 
    instance $\instance$ with submodular $f$, we have $\maxreward(b) \geq (1/\gaprewardsubmodular(b,B)) \cdot \maxreward(B)$
for every $b$ with $\max_{i \in \agents} p(\{i\}) \leq b \leq B$. 
Moreover, this argument extends to any class of instances (not just those with submodular $f$) and any objective (not just the reward).
    }
\end{figure}

Next, we present a tight characterization of the price of frugality for reward and welfare when $f$ is submodular.

\begin{restatable}[Price of Frugality with Submodular $f$]
    {theorem}{frugalitysubmodularthm} \label{thm:main_price}
 For any $0 < b \leq B \leq 1$, it holds that:
    \begin{align*}
        \gaprewardsubmodular (b,B) =  \min\left(\lceil 2B/b \rceil - 1,n \right) \leq \min\left({2} B/b , n \right).
    \end{align*}
    Furthermore, for submodular instances, it holds that the worst-case price of frugality for welfare is equal to the worst-case price of frugality for reward, i.e.,   $\gapwelfaresubmodular (b,B) 
 = \min\left(\lceil 2B/b \rceil - 1,n \right)$.
\end{restatable}

The lemma used to derive the lower bounds on the price of frugality for reward and welfare (\Cref{lem:pof_upper}) is proved below, with the corresponding instance shown in \Cref{fig:impo}.
The bounds on the price of frugality for submodular reward are depicted in \Cref{fig:price_of_frugality}.

Let us now explain an important connection between the price of frugality and the downsizing lemmas (\Cref{lem:submodular_payment_scaling,lem:scale_xos}).

\begin{remark}[Optimality of the Downsizing Lemma for Submodular Reward]\label{rem:downsizing_opt}
    The upper bounds on the price of frugality for reward under both submodular and XOS $f$ follow from the corresponding downsizing lemmas for submodular and XOS $f$, respectively; see \Cref{lem:positive_pof_subadd,lem:upperxosfdlkdsfj}.

Since the bound obtained via the downsizing lemma for submodular $f$ is tight, this lemma is optimal for any target budget $0 < b < B$. Specifically, for any subset $S \subseteq \agents$ with $p(S) = B$, the optimal way to choose a subset $T \subseteq S$ satisfying either $p(T) \leq b$ or $|T| = 1$ is to apply \Cref{lem:submodular_payment_scaling} with the smallest parameter $M$ such that $2/M \leq b/B$. This optimality holds even when $f$ is restricted to being additive.
\end{remark}

We first bound the price of frugality for any subadditive objective, including reward and welfare, for instances where $f$ is submodular.

\begin{restatable}[Upper Bound on Price of Frugality with Submodular $f$]
    {lemma}{pofuppersubadditivesubmodular} \label{lem:positive_pof_subadd}
Fix any $0 < b < B$ and any subadditive objective $\varphi$.
    For any 
    instance $\instance$ with submodular $f$ and $p(\{i\}) \leq b$ for all agents $i \in \agents$, it holds that:
    \begin{align*}
   \gapphi(b,B) \leq \min\left(\lceil 2B/b \rceil - 1,n \right).
    \end{align*}
\end{restatable}

\begin{proof}
Let $S \subseteq \agents$ be the team with $p(S) \leq B$ that maximizes $\varphi(S)$, i.e., $\varphi(S) = \MAX\varphi(B)$.  

First, we show that $\MAX\varphi(b) \geq (1/n) \cdot \MAX\varphi(B)$. Since $\varphi$ is subadditive, we have:
\[
\MAX\varphi(B) = \varphi(S) \leq \sum_{i \in S} \varphi(\{i\}) \leq (1/n) \cdot \max_{i \in \agents} \varphi(\{i\}) \leq (1/n) \cdot \MAX\varphi(b),
\]  
where the last step holds because $p(\{i\}) \leq b$ for all agents $i \in \agents$ by assumption.  

Now, define $M = \lceil 2B/b \rceil$. We show that $\MAX\varphi(b) \geq (1/(M-1)) \cdot \MAX\varphi(B)$. Since $b < B$ implies $2B/b > 2$, we have $M \geq 3$.  
Thus, applying \Cref{lem:submodular_payment_scaling}, there exists a subset $T \subseteq S$ such that $\varphi(T) \geq (1/(M-1)) \cdot \varphi(S)$ and either $p(T) \leq (2/M) \cdot p(S)$ or $|T| = 1$.  

If $|T| = 1$, then $p(T) \leq b$ by assumption. Otherwise, since $M = \lceil 2B/b \rceil$, it follows that $p(T) \leq (2/M) \cdot p(S) \leq (b/B) \cdot B = b$. Thus, $p(T) \leq b$, which concludes the proof.
\end{proof}

We now establish a lower bound on the price of frugality for both reward and welfare, which holds even when $f$ is additive. A matching upper bound is provided in \Cref{lem:positive_pof_subadd}.

\begin{restatable}[Lower Bound on Price of Frugality with Submodular $f$]
    {lemma}{poflowersubadditivesubmodular} \label{lem:pof_upper}
Fix any $0 < b < B \leq 1$.
    There exists an
    instance $\instance$ with additive $f$ and $p(\{i\}) \leq b$ for all agents $i \in \agents$ where:
    \begin{align*}
    \gapreward(b,B) =  \gapwelfare(b,B) = \min\left(\lceil 2B/b \rceil - 1,n \right).
    \end{align*}
\end{restatable}

\begin{proof}
Let $M = \min(\lceil 2B/b \rceil - 1,n)$. Note that $M < 2B/b$. 
Consider an instance with $n$ agents, where the reward 
function is defined as $f(S) = (1/M) \cdot |S \cap \{1, \ldots, M\}|$ and the cost is set to $c_i = B/M^2$ for all $i \in \agents$. See \Cref{fig:impo} for an illustration of this instance.

For any $S \subseteq [n]$,
we have $c_i / f_S(i) = B/M$ for $i \in \{1, \ldots, M\}$, and $c_i / f_S(i) = \infty$ for $i > M$.
Therefore, letting $S = \{1, \ldots, M\}$, we obtain $p(S) = \sum_{i \in S} c_i / f_{S}(i) = B$. Additionally, $f(S) = 1$ and $c(S) = B/M$. This implies that $\maxreward(B) \geq 1$ and $\maxwelfare(B) \geq 1 - B/M$.

Now, consider any incentivizable set $T \subseteq \agents$ with $p(T) \leq b$. Since $T$ must be contained within $\{1, \ldots, M\}$ to be incentivizable, we get $p(T) = |T| \cdot (B/M) \leq b$. Since $M < 2B/b$, we have $|T| \leq 1$. 

For every $i \in \{1, \ldots, M\}$, we observe that $f(\{i\}) = 1/M$ and $c(\{i\}) = B/M^2$. Hence, $\maxreward(b) \leq 1/M$ and $\maxwelfare(b) \leq (1/M) \cdot (1 - B/M)$. The result follows.
\end{proof}

Finally, \Cref{thm:main_price} follows directly from \Cref{lem:positive_pof_subadd,lem:pof_upper}.

\subsection{Price of Frugality in XOS and Subadditive Instances}\label{sec:pof_xossubadd}

We also establish a separation between the price of frugality for submodular and XOS reward functions.  
Recall that $\gaprewardsubmodular(b,B) \leq 2$ for any $b \in [(2/3) \cdot B, B)$ by \Cref{lem:positive_pof_subadd}. The following lemma shows that this inequality does not hold for XOS reward functions. 
Consequently, the guarantees of the downsizing algorithm for submodular $f$ (\Cref{lem:submodular_payment_scaling}) cannot be achieved for XOS $f$, necessitating weaker guarantees for our downsizing algorithm for XOS (\Cref{lem:scale_xos}).

\begin{restatable}[Lower Bound on Price of Frugality with XOS $f$]
    {lemma}{hardpofxos} \label{lem:hard_pof_xos}
For any $0 < b < B$,
    there is an
    instance $\instance$ with XOS $f$ and $p(\{i\}) \leq b$ for all  $i \in \agents$ where:
\begin{align*}
 \gapreward(b,B) \geq 5/2.
\end{align*}
\end{restatable}
\begin{proof}
Since $\gapreward(b,B)$ increases with $B$, it suffices to show that $\gapreward(b,B) \geq 5/2$ for sufficiently small $B > b$. Hence, we assume $B \leq 2b$.

We define an instance with three agents $\agents = \{1,2,3\}$ and a reward function $f(S) = \max(a_1(S), a_2(S))$, where $a_1$ and $a_2$ are additive functions with values $a_1(1) = 2/5$, $a_1(2) = 2/5$, $a_1(3) = 1/5$, and $a_2(1) = 0$, $a_2(2) = 0$, $a_2(3) = 2/5$. The agents' costs are given as $c_1 = B/5$, $c_2 = B/5$, and $c_3 = 0$. 

We obtain $f(\{1,2,3\}) = a_1(\{1,2,3\}) = 1$. The corresponding payment is $p(\{1,2,3\}) = c_1 / f_{\{1,2,3\}}(1) + c_2 / f_{\{1,2,3\}}(2) = (B/5)/(2/5) + (B/5)/(2/5) = B$. Thus, we have $\maxreward(B) \geq 1$.

Next, we consider pairs of agents. We compute $p(\{1,2\}) = c_1 / f_{\{1,2\}}(1) + c_2/ f_{\{1,2\}}(2) = (B/5)/(2/5) + (B/5)/(2/5) = B$. Moreover, $p(\{1,3\}) = c_1 / f_{\{1,3\}}(1) = (B/5)/(1/5) = B$ and, by symmetry, $p(\{2,3\}) = B$. Therefore, all feasible subsets with respect to budget $b$ are singletons.

Considering singletons, we $p(\{1\}) = c_1/f_{\{1\}}(1) = (B/5)/(2/5) = B/2 \leq b$, where the inequality follows from our assumption on $B$. Also, $f({1}) = a_1(1) = 2/5$.
Similarly, $p(\{2\}) = c_2/f_{\{2\}}(2) = (B/5)/(2/5) = B/2 \le b$ and $f(\{2\}) = a_1(2) = 2/5$. For agent $3$, we get $p(\{3\}) = c_3/f_{\{3\}}(3) = 0 \leq b$ and $f(\{3\}) = a_2(3) = 2/5$.
Thus, $\maxreward(b) \leq 2/5$, and the result follows. 
\end{proof}

Finally, we establish a lower bound on the price of frugality for instances with subadditive \( f \), demonstrating a separation between instances with XOS \( f \) and those with subadditive \( f \). While for XOS \( f \), the price of frugality scales as \(O( B/b ) \), for subadditive \( f \), it grows as \( \Omega(\sqrt{n}) \) even when \( b \) is arbitrarily close to \( B \). 
The proof of the next lemma follows the approach used in \cite[Theorem 4.1]{duetting2022multi}.
{As a result, no downsizing algorithm can achieve a constant-factor guarantee for subadditive $f$.}

\begin{restatable}[Lower Bound on Price of Frugality with Subadditive $f$]
    {lemma}{pofsubadditive} \label{lem:pof_subadditive}
For any $0 < b < B \leq 1$, there exists an
instance $\instance$ with subadditive $f$ and $p(\{i\}) \leq b$ for all  $i \in \agents$ such that:
    \begin{align*}
    \gapreward(b,B) = \Omega(\sqrt{n}).
    \end{align*}
\end{restatable}
\begin{proof}
   Assume without loss of generality that $n \geq 4$ and that $n$ is even.
   Since $\gapreward(b,B)$ is increasing in $B$, it suffices to prove the statement for sufficiently small $B > b$. Thus, given that $nb/2 > b$, we assume $B \leq nb/2$.
Consider an instance with $n$ agents, where the cost for each agent is given by  $c_i = B / ((n/2 + 1) \cdot \sqrt{n})$ for all $i \in \agents$. The reward function is defined as:
    \begin{align*}
        f(S) = \begin{cases}
             1/\sqrt{n} + |S|/n & \text{if }  |S| \leq n/2\\
             2/\sqrt{n} + 1/2 & \text{if } |S| \geq n/2 + 1
        \end{cases}
    \end{align*}
    
    Let us first argue that \( f \) is subadditive. We need to verify that for \( S_1, S_2 \subseteq \agents \), we have \( f(S_1 \cup S_2) \leq f(S_1) + f(S_2) \). Note that \( f(S) \leq 2/\sqrt{n} + |S|/n \) for every \( S \subseteq \agents \).  
    If \( |S_1| \geq n/2+1 \) or \( |S_2| \geq n/2+1 \), then we have \( f(S_1 \cup S_2) = f(S_1) \) or \(  f(S_1 \cup S_2) = f(S_2)\), respectively.
    If \( |S_1| \leq n/2 \) and \( |S_2| \leq n/2 \), then
    \begin{align*}
        f(S_1 \cup S_2) \leq 2/\sqrt{n} + |S_1 \cup S_2|/n \leq (1/\sqrt{n} + |S_1|/n) + (1/\sqrt{n} + |S_2|/n) = f(S_1) + f(S_2).
    \end{align*}
     Therefore, $f$ is subadditive.

    Let us also argue that $p(\{i\}) \leq b$ for all agents $i \in \agents$. Note that $f_{\{i\}}(i) = 1/\sqrt{n} + 1/n \geq 1/\sqrt{n}$, and so $p(\{i\}) = c_i / f_{\{i\}}(i) \leq \sqrt{n} \cdot B/((n/2+1) \cdot \sqrt{n}) = B/(n/2+1) \leq 2B/n \leq b$ by the assumption on $B$. 
    
    For \( |S| = n/2 + 1 \), we have \( f(S) = 2/\sqrt{n} + 1/2 \) and \( f(S \setminus \{i\}) = 1/\sqrt{n} + (n/2)/n = 1/\sqrt{n} + 1/2 \), so \( f_S(i) = 1/\sqrt{n} \) for all \( i \in S \). Therefore, \( p(S) = \sum_{i \in S} c_i / f_S(i) = |S| \cdot c_i / f_S(i) = (n/2+1) \cdot c_i / (1/\sqrt{n}) = B \) by our choice of \( c_i \). Thus, \( \maxreward(B) \geq 2/\sqrt{n} + 1/2 \).

    Note that for $|S| = 1$, we have $f(S) \leq 1/\sqrt{n} + 1/n \leq 2/\sqrt{n}$.
      Moreover,  for \( 2 \leq |S| \leq n/2 \), we have \( f_S(i) = 1/n \) for all \( i \in S \), and so if \( p(S) \leq b < B \), we must have:
      \begin{align*}
          p(S) = |S| \cdot \frac{c_i}{f_S(i)} = |S| \cdot \frac{B / ((n/2 + 1) \cdot \sqrt{n})}{1/n} = |S| \cdot \frac{n}{n/2 + 1} \cdot \frac{1}{\sqrt{n}} \cdot B < B \quad \Longrightarrow \quad |S| < \sqrt{n}.
      \end{align*} 
      This implies that \( \maxreward(b) \leq f(\{1, \ldots, \lfloor \sqrt{n} \rfloor \}) \leq 1/\sqrt{n} + (1/\sqrt{n})/n \leq 2/\sqrt{n} \). The result follows.
\end{proof}

\subsection{Price of Frugality for Profit}\label{sec:pof_revene}

Next, we present refined bounds on the price of frugality for profit when $f$ is submodular. 
The bounds in the following theorem are tight up to a constant factor.

\begin{restatable}[Price of Frugality for Profit]
    {theorem}{frugalitysubmodularthmrevenue} \label{thm:revenuepofff}
 For any $0 < b \leq B \leq 1$, it holds that:
    \begin{align*}
         \max\left( 2-b, k \cdot \left({2-k\cdot b}\right)/\left({2-b}\right) \right) \leq \gaprevenuesubmodular(b,B) \leq \min\left(\lceil 2B/b \rceil - 1,n \right) 
    \end{align*}
    where $k = \min(\lfloor 1 / b + 1/2 \rfloor, \lceil 2B/b \rceil - 1, n)$. 
\end{restatable}

In this section, we provide only the upper bound (\Cref{lem:revenue_upper}). The full proof of \Cref{thm:revenuepofff}, including the lower bound, is deferred to \Cref{sec:price_of_f_proofs}.
Note that for the lower bound in \Cref{thm:revenuepofff}, the first term ($2-b$) dominates the second ($ k \cdot \left({2-k\cdot b}\right)/\left({2-b}\right) $) if and only if \( b > (\sqrt{33}-5)/2 \approx 0.372 \) or \( k \leq 2 \).

While our bounds for the price of frugality for reward and welfare are tight (\Cref{lem:positive_pof_subadd,lem:pof_upper}), a gap remains in our bounds for the price of frugality for profit (\Cref{thm:revenuepofff}).  We propose the following conjecture regarding the price of frugality for profit.

\begin{conjecture}
    For any $0 < b \leq B \leq 1$, let $k = \min(\lfloor 1 / b + 1/2 \rfloor, \lceil 2B/b \rceil - 1, n)$.
    We conjecture that 
$\gaprevenuesubmodular(b,B) = 
        \max\left( 2-b, k \cdot \left({2-k\cdot b}\right)/\left({2-b}\right)  \right)$.
\end{conjecture}

We now show that any upper bound on the price of frugality for reward also applies to the price of frugality for profit. Notably, this result holds without any assumptions on $f$.

\begin{restatable}[Upper Bound on Price of Frugality for Profit]
    {lemma}{pofrevenuesubmodularupper} \label{lem:revenue_upper}
For any $0 < b < B$, any
instance $\instance$ with monotone $f$ and  $p(\{i\}) \leq b$ for all agents $i \in \agents$, it holds that:
    \begin{align*}
    \gaprevenue(b,B) \leq  \gapreward(b,B).
    \end{align*}
\end{restatable}

\begin{proof}
Let $S_B \subseteq \agents$ be the profit-maximizing team satisfying $\maxrevenue(B) = f(S_B)$ and $p(S_B) \leq B$. If $p(S_B) \leq b$, then $\maxrevenue(b) = \maxrevenue(B)$, proving the result. Otherwise, let $S_b \subseteq \agents$ be the reward-maximizing team satisfying $\maxreward(b) = f(S_b)$ and $p(S_b) \leq b$. Then, it holds that $\maxrevenue(b) \geq (1 - p(S_b)) \cdot f(S_b) \geq (1 - p(S_B)) \cdot f(S_b) = (1 - p(S_B)) \cdot \maxreward(b)$. Moreover, $\maxrevenue(B) = (1 - p(S_B)) \cdot f(S_B) \leq (1 - p(S_B)) \cdot \maxreward(B)$, which completes the proof.
\end{proof}

\newpage
\bibliographystyle{alpha}
\bibliography{refs}

\appendix

\section{Tighter Guarantees for Submodular Instances} \label{app:objective_reductions_submodular}
In this section we show versions of  \Cref{lem:reduction_to_mrl} and \Cref{lem:reduction_from_mrl} with improved constants for the case where $f$ is submodular. More specifically, we prove the following two lemmas:
\begin{lemma} 
[Constant-Factor Reduction from $\MAX\varphi(B)$ to $\maxlightreward(B)$, submodular $f$]
\label{lem:reduction_to_mrl_submodular}
    Fix an instance $\instance$ with submodular $f$, a budget $B\in (0,1]$, and a \goodobj  \
    objective $\varphi$.
    For any given team $S \subseteq \agents$ that is a $\gamma$-approximation to $\maxlightreward(B)$, let $S'$ be the result of applying \Cref{lem:submodular_payment_scaling} to $S$ with $M=3$. Then, it holds that one of $\{\{i\}\}_{i\in \agents} \cup \{S\}$ is a $(6\gamma+1)$-approximation to $\MAX\varphi(B)$.
\end{lemma}

\begin{lemma}
[Constant-Factor Reduction from $\maxlightreward(B)$ to $\MAX\varphi(B')$, submodular $f$]
 \label{lem:reduction_from_mrl_submodular}
    Let $I=\instance$  be an
    instance such that $f$ is submodular, let $\varphi$ be a \goodobj\  objective, and let $B, B'\in(0,1]$.
    Consider the instance $I'=\langle \lightagents, f\mid_\lightagents, \frac{c\cdot B'}{B}\mid_\lightagents\rangle$, i.e., $I'$ is the same as $I$ except with only the light agents, and all costs have been scaled by $\frac{B'}{B}$. 
    
    If $S$ is a $\gamma$-approximation to $\MAX\varphi(B')$ in $I'$, then one of $\{\{i\}\}_{i\in \agents} \cup \{S\}$ is a $6\gamma$-approximation to $\maxlightreward(B)$ in $I$.
\end{lemma}

We first prove a version of \Cref{lemma:goodobj_upper_bound} for the case of submodular $f$ with a better constant:
\begin{lemma} [Key Property of \goodobj\ Objectives with Submodular $f$]\label{lemma:goodobj_upper_bound_submodular}
 Fix an instance $\instance$ with submodular $f$, a budget $B\in (0,1]$, and a \goodobj\ objective $\varphi$. It holds that
\[
\MAX\varphi(B) \le \maxlightreward(B)+\max_{i\in \agents} \varphi(\{i\}).
\]
\end{lemma}
\begin{proof}
    Let $S^\star$ be an optimal solution to $\MAX\varphi(B)$, i.e., $\varphi(S^\star) = \MAX\varphi(B)$ and $p(S^\star) \le B$.
    Note that, by \Cref{obs:1_heavy}, $S^\star$ contains at most $1$ heavy agent. If it contains no heavy agents, we have 
    $\varphi(S^\star) \le f(S^\star) \le \maxlightreward(B)$, as needed.
    Otherwise, let $S^\star \setminus \lightagents= \{i^\star\}$ and we get 
     
    \begin{align*}
    \varphi(S^\star)
    &\le
    \varphi(S^\star \cap \lightagents) + \varphi(\{i^\star\}) && (\text{property $1$ of \goodobj\ objectives}) \\
    &\le
    f(S^\star \cap \lightagents) + \max_{i \in A} \varphi(\{i\}) && (\varphi \le f).  \\
    &\le \maxlightreward(B)+\max_{i\in \agents} \varphi(\{i\}) && (P(S^\star\cap L) \le P(S^\star)\le B),
    \end{align*}
    where the last inequality is because when $f$ is submodular the payment function $p$ is monotone.
\end{proof}

\begin{proof} [Proof of \Cref{lem:reduction_to_mrl_submodular}]
    We start by showing that 
    \[
    6\gamma \cdot \varphi(S') \ge \maxlightreward(B).
    \]
    This holds, because by the guarantees of \Cref{lem:submodular_payment_scaling} we have 
    \[f(S')\ge \frac{1}{M-1} f(S^\star) \ge \frac{1}{2\gamma} \maxlightreward(B),
    \]
    Moreover, we have either $|S'| = 1$ or $p(S') \leq (2/3) \cdot p(S)$. Observe that $S'$ is budget-feasible since, if $|S'|\ne 1$ then $p(S')\le \frac{2}{3}p(S)\le \frac{2}{3}B\le B$, and if $|S'|=1$ then from sub-additivity of $f$,
    \[
    p(\{i\})=\frac{c_i}{f(\{i\})}\le \frac{c_i}{f(S)-f(S\setminus\{i\})}=\frac{c_i}{f_S(i)}\le \sum_{j\in S}\frac{c_j}{f_S(j)}=p(S)\le B.
    \]
    
    We will now show that $p(S') \leq 2/3$. If $|S'| = 1$, then since $S' \subseteq \lightagents$, it follows that $p(S') \leq 1/2$.
    If $p(S') \leq (2/3) \cdot p(S)$, then by the budget feasibility of $S$, we get $p(S') \leq (2/3) \cdot p(S) \leq (2/3) \cdot B \le 2/3$. 
 Now,
    \[
    \varphi(S') \ge g(S') = (1-p(S'))f(S') \ge \left(1-\frac{2}{3}\right)\frac{1}{2} f(S) = \frac{1}{6\gamma}\maxlightreward(B).
    \]
    Thus, by \Cref{lemma:goodobj_upper_bound_submodular} we have
    \begin{align*}
    \MAX\varphi(B)&\le \maxlightreward(B)+\max_{i\in \agents} \varphi(\{i\}) \le 6\gamma \cdot \varphi(S')+\max_{i\in \agents} \varphi(\{i\}),
    \end{align*}
    as needed.
\end{proof}

\begin{proof} [Proof of \Cref{lem:reduction_from_mrl_submodular}]
Let $S^\star$ be a solution to $\maxlightreward(B)$, i.e., $S\subseteq \lightagents$ and satisfies both $f(S^\star) = \maxlightreward(B)$ and $p(S^\star) \le B$. 
Apply \Cref{lem:submodular_payment_scaling} to $S^\star$ with $M=3$, and get a set $T$ such that $f(T) \ge \frac{1}{M-1} f(S^\star)=\frac{1}{2}f(S^\star)$ and either $|T|=1$ or $p(T) \le \frac{2}{3} p(S) \le \frac{2}{3}B$. Observe that, in either case, $p(T)\le p(S^\star)\le B$.

Note that if $|T|=1$, then $T$ itself is a $2$-approximation to $\maxlightreward(B)$, and therefore also a $6\gamma$ approximation, concluding the proof. 

Otherwise, let $p'$, $\varphi'$, and $g'$ denote the total payment, objective $\varphi$, and the principal's profit (respectively) in the scaled instance $I'$. Note that $p'(T) =\frac{B'}{B}p(T) \le \frac{2}{3}B' \le \frac{2}{3}$, and therefore $T$ is budget-feasible in $I'$ with respect to $B'$, and also
\[
\begin{split}
f(S) &\ge \varphi'(S) \ge \frac{1}{\gamma} \varphi'(T) \ge \frac{1}{\gamma} g'(T) =\frac{1}{\gamma} (1-p'(T))f(T) \\
&\ge \frac{1}{\gamma}\left(1-\frac{2}{3}\right)\frac{1}{2}f(S^\star) \ge \frac{1}{6\gamma} \maxlightreward(B).
\end{split}
\]
Additionally, $S$ is budget-feasible in the original instance $I$, since $p(S) = \frac{B}{B'}\cdot p'(S) \le \frac{B}{B'} \cdot B'=B$, concluding the proof.
\end{proof}

\section{FPTAS for Additive Instances}
\label{app:additive_fptas}

In this section we consider the multi-agent 
budgeted setting when $f$ is additive. Notably, in this case any additive objective (and in particular the expected reward and social welfare) is equivalent to a \textsc{Knapsack} problem, implying that exact solutions are NP-hard to compute and the existence of FPTAS.
\begin{remark}\label{remark:additive_f_sw_reward_fptas}
    For the case of additive $f$, let $\varphi$ be an additive objective (i.e., for any instance $\instance$ it holds that $\varphi_{\instance}$ is additive) and let $B\in (0,1]$. The problem $\MAX\varphi(B)$ is equivalent to the knapsack problem with items $\agents$, capacity $B$, weights $\frac{c_i}{f(\{i\})}$ for all $i\in \agents$, and values $\varphi(\{i\})$.
\end{remark}
\cite{duetting2022multi} show that in a multi-agent setting with 
an additive reward function $f$, the optimal contract problem is NP-hard. This hardness carries over trivially to the corresponding budgeted setting. Additionally, \cite{duetting2022multi} demonstrate that this problem admits an FPTAS. The following proposition establishes that the algorithms used to achieve this result can be adapted to provide an FPTAS for budgeted settings as well.

\newcommand{\tf}{\tilde{f}}

\begin{proposition}\label{prop:FPTAS}
    In a multi-agent setting with additive reward function $f$, the optimal budgeted contract problem admits an FPTAS.    
\end{proposition}

Let $S^*$ be the set that maximizes the principal's utility under budget constraints., with $b = \max_{i \in S^*} f_i$. We may assume we know $b$, as we can run the algorithm with the $n$ different possible values. 

For any $\eps >0$, let $\delta = \frac{\eps}{n}$.
We define a rounded reward function as follows: $\tf(\{i\}) = \floor{\frac{f(\{i\})}{\delta b}}\cdot \delta b$ and $\tf(S) = \sum_{i \in S} \tf(\{i\})$. 
Observe that all values of $\tf$ are multiples of $\delta b$.

Let $T: \{0,\dots,\ceil{\frac{n}{\delta}}\} \to 2^{[n]}$ be the function that for every $k \in \{0,\dots,\ceil{\frac{n}{\delta}}\}$ returns the set $T(k)$ that minimizes $\sum_{i \in S} \frac{c_i}{f(\{i\})}$ subject to $\tf(S) \ge k\delta b$.

It was shown in \cite{duetting2022multi} that the table which represent the function $T$ can be computed in poly-time in $n$ and $\frac{1}{\eps}$.

\begin{proof}[Proof of \cref{prop:FPTAS}]

    Our algorithm returns the set of agents $T^*$ that maximizes the principal's profit among all sets computed in the table $T$ which are budget-feasible. Namely, $T^* := T(k^*)$, is the budget-feasible set, i.e. $\sum_{i \in T(k^*)} \frac{c_i}{f(\{i\})} \le B$, which maximizes $\left( 1 -  \sum_{i \in S} \frac{c_i}{f(\{i\})} \right) k \delta b$ for any $k$ as above.
    
    We will show that $g(T^*) \ge (1-\eps) g(S^*)$:
    \begin{eqnarray*}
        g(T^*)
        =
        \left( 1 -  \sum_{i \in T^*} \frac{c_i}{f(\{i\})} \right)f(T^*) 
        \ge
        \left( 1 -  \sum_{i \in T^*} \frac{c_i}{f(\{i\})} \right) k^* \delta b
    \end{eqnarray*}

    Denote $\tf(S^*) = \hat{k} \delta b$. 
    Since $S^*$ is budget-feasible and $f(S^*) \ge \tf(S^*) = \hat{k} \delta b$, we get, by definition of $T$, that $T(\hat{k})$ is also budget feasible. By optimality of $T^*$ we have
    $$
    \left( 1 -  \sum_{i \in T^*} \frac{c_i}{f(\{i\})} \right) k^* \delta b
    \ge
    \left( 1 -  \sum_{i \in T(\hat{k})} \frac{c_i}{f(\{i\})} \right)\hat{k}\delta b
    \ge
    \left( 1 -  \sum_{i \in S^*} \frac{c_i}{f(\{i\})} \right)\tf(S^*)
    $$
    where the last inequality is by definition of the function $T$.

    Finally, observe that
    $$
    \tf(S^*) = \sum_{i \in S^*} \tf(\{i\}) 
    \ge \sum_{i \in S^*} (f(\{i\}) - \delta b)
    \ge f(S^*) - n\delta b
    = f(S^*) - \eps \cdot \max_{i \in S^*} f(\{i\})
    \ge (1-\eps)f(S^*)
    $$
    and we have
    $g(T^*) \ge (1-\eps) g(S^*)$.
\end{proof}

\section{Missing Proofs from Section 5} \label{sec:price_of_f_proofs}

\subsection{Proof of Theorem 5.2}

In the following lemma, we upper bound  the price of frugality for any BEST objective under XOS $f$.

\begin{lemma}[Upper Bound on Price of Frugality for BEST Objectives]\label{lem:upperxosfdlkdsfj}
     For any $0 < b < B \leq 1$ and any \goodobj\ objective $\varphi$, it holds that:
\begin{align*}
    \gapphixos(b,B) = O(\min(B/b, n))
\end{align*}
\end{lemma}

\begin{proof}
From \Cref{lemma:goodobj_upper_bound} we know that
\[
\frac{\MAX\varphi(B)}{\MAX\varphi(b)}\le \frac{2 \cdot \maxlightreward(B) +\max_{i\in \agents} \varphi(\{i\})}{\MAX\varphi(b)}.
\]
First, we note that by the assumption that any single agent is incentivizable under $b$, it holds that $\MAX\varphi(b)\ge \max_{i\in \agents}\varphi(\{i\})$. Therefore, for asymptotic bounds we are interested in the ratio
\[
\frac{2\cdot \maxlightreward(B)}{\MAX\varphi(b)}.
\]
Note that because $f$ is subadditive we have 
\[
\begin{split}
2\maxlightreward(B) &\le n\cdot \max_{i\in \lightagents} f(\{i\}) \le n\cdot 2\max_{i\in \lightagents} (1-p(\{i\}) \cdot f(\{i\}) \le 2n\cdot \max_{i\in \lightagents} g(\{i\}) \\
&\le \max_{i\in \lightagents} \varphi(\{i\}) \le 2n\cdot \MAX\varphi(b).
\end{split}
\]

Now, let $S\subseteq \lightagents$ be a solution to $\maxlightreward(B)$, and define $M = \lceil 8B/b \rceil$. We show that $\MAX\varphi(b) \geq (1/(4M-4)) \cdot \maxlightreward(B)$. Since $b < B$ implies $8B/b > 8$, we have $M \geq 3$.  
Thus, applying \Cref{lem:scale_xos}, there exists a subset $T \subseteq S$ such that $\varphi(T) \geq (1/(2M-2)) \cdot \varphi(S)$ and either $p(T) \leq (4/M) \cdot p(S)$ or $|T| = 1$.  
If $|T| = 1$, then $p(T) \leq b$ by assumption, and because $S\subseteq\lightagents$ we also have $p(T)\le \frac{1}{2}$. Otherwise, since $M = \lceil 8B/b \rceil$, it follows that $p(T) \leq (4/M) \cdot p(S) \leq (b/2B) \cdot B = \frac{1}{2}b$. Thus, $p(T) \leq \frac{1}{2}b$. In particular, $T$ is budget-feasible with respect to $b$ and $p(T)\le \frac{1}{2}$, so
\[
\varphi(T) \ge g(T) =(1-p(T)) \cdot f(T) \ge \frac{1}{2} \frac{1}{(2M-2)}f(S)=\frac{1}{(4M-4)}\maxlightreward(B).
\]

This concludes the proof since $1/(4M-4) = 1/(4\lceil 8B/b \rceil -4) \geq 1/(4(8B/b + 1) - 4) = (1/32) \cdot (b/B)$.
\end{proof}

We also give an asymptotically tight lower bound.

\begin{restatable}[Lower Bound on Price of Frugality for BEST Objectives]
    {lemma}{pofsandwiched} \label{lem:pof_sandwiched}
 For any $0 < b < B \leq 1$ and any \goodobj\ objective $\varphi$, it holds that:
\begin{align*}
 \gapphixos(b,B) = \Omega(\min(B/b,n)).
\end{align*}
\end{restatable}
\begin{proof}
    If $B \leq 2b$, then $\min(B/b, n) = \Omega(1)$, so it suffices to observe that $\gapphixos(b, B) \geq 1$. For $B > 2b$, we have:  
    \begin{align*}
        \MAX\varphi(B) \geq \maxrevenue(B) \geq (1-B/2) \cdot \maxreward(B/2) \geq (1/2) \cdot \maxreward(B/2)
    \end{align*}
   and $\MAX\varphi(b) \leq \maxreward(b)$.  
Therefore:
\begin{align*}
    \gapphixos(b, B) \geq \gaprewardxos(b, B/2) = \Omega(\min(B/(2b), n)).
\end{align*}  
The result follows.  
\end{proof}

We are finally ready to prove \Cref{thm:price_asymptotic}.

\asymptoticfrugality*
\begin{proof}
This follows directly from \Cref{lem:upperxosfdlkdsfj,lem:pof_sandwiched}.
\end{proof}

\subsection{Proof of Theorem 5.9}

In the following lemma, we prove the first lower bound on the price of frugality for profit.

\begin{figure}[t]
    \centering
    \begin{subfigure}{0.48\textwidth}
        \centering
        \begin{tikzpicture}
    \begin{axis}[
        width=\myheight, height=\myheight,
        xlabel={payment $p$},
        ylabel={},
        xmin=0, xmax=1,
        ymin=0, ymax=1,
        grid=none, 
        legend pos=south east,
        clip=false, 
        xtick={0, 0.4,   1},
        xticklabels={$0$,  $0.4$,  $1$},
        ytick={0, 0.285, 0.44, 0.5, 0.585, 0.8, 1},
        yticklabels={$0$, $0.3$, $0.48$, $0.5$, $0.6$, $0.8$, $1$},
        tick label style={font=\small},
    ]

\addplot[black, thick] coordinates {(0.4,0) (0.4,1)};

    \addplot[red, thick] coordinates {(0,0.01) (0.05,0.01)};
    \addplot[red, only marks, mark=*, mark size=1pt] coordinates {(0,0.01) (0.05,0.3) (0.4, 0.5) (0.45, 0.8) (1,0.8)};
    \addplot[red, only marks, mark=o, mark size=1pt] coordinates {(0.05,0.01) (0.4,0.3) (0.45,0.5)};
    \addplot[red, thick] coordinates {(0.05,0.3) (0.4,0.3)};
    \addplot[red, thick] coordinates {(0.4,0.5) (0.45,0.5)};
    \addplot[red, thick] coordinates {(0.45,0.8) (1,0.8)};

    \addplot[blue, thick] coordinates {(0,0.0) (0.05,0.00)};
    \addplot[blue, thick] coordinates {(0.05,0.285) (0.4,0.285)};
    \addplot[blue, thick] coordinates {(0.4,0.3) (0.45,0.3)};
    \addplot[blue, thick] coordinates {(0.45,0.585) (1,0.585)};
    \addplot[blue, only marks, mark=*, mark size=1pt] coordinates {(0,0.0) (0.05,0.285) (0.4, 0.3) (0.45, 0.585) (1,0.585)};
    \addplot[blue, only marks, mark=o, mark size=1pt] coordinates {(0.05,0.0) (0.4,0.285) (0.45,0.3)};

    \addplot[green!60!black, thick] coordinates {(0.05, 0.285) (0.4, 0.18)};
    \addplot[green!60!black, thick] coordinates {(0.4, 0.3) (0.45, 0.27)};
    \addplot[green!60!black, thick] coordinates {(0.45, 0.44) (1, 0)};

    \addplot[green!60!black, only marks, mark=*, mark size=1pt] coordinates {(0.05, 0.285) (0.4, 0.3) (0.45, 0.44) (1,0)};
    \addplot[green!60!black, only marks, mark=o, mark size=1pt] coordinates {(0.4, 0.18) (0.45, 0.27)};

    \addplot[fill=red, fill opacity=0.2, draw=none] 
        coordinates {(0.05,0.285) (0.4,0.285) (0.4, 0.3) (0.05, 0.3)};
    \addplot[fill=red, fill opacity=0.2, draw=none] 
        coordinates {(0.4,0.3) (0.45,0.3) (0.45, 0.5) (0.4, 0.5)};
    \addplot[fill=red, fill opacity=0.2, draw=none] 
        coordinates {(0.45,0.585) (1,0.585) (1, 0.8) (0.45, 0.8)};

    \node at (0.3, 0.93) {$p \leq b$};
    \node at (0.49, 0.93) {$p > b$};

         \addplot[fill=blue, fill opacity=0.2, draw=none] 
        coordinates {(0.05,0.285) (0.4,0.285) (0.4, 0.18) (0.05, 0.285)};
    \addplot[fill=blue, fill opacity=0.2, draw=none] 
        coordinates {(0.4,0.3) (0.45,0.3) (0.45, 0.27) (0.4, 0.3)};
    \addplot[fill=blue, fill opacity=0.2, draw=none] 
         coordinates {(0.45,0.585) (1,0.585) (1, 0) (0.45, 0.44)};

    \addplot[fill=green!60!black, fill opacity=0.2, draw=none] 
        coordinates {(0.05,0.285) (0.4,0.18) (0.4, 0) (0.05, 0)};
    \addplot[fill=green!60!black, fill opacity=0.2, draw=none] 
        coordinates {(0.4,0.3) (0.45,0.27) (0.45, 0) (0.4, 0)};
    \addplot[fill=green!60!black, fill opacity=0.2, draw=none] 
        coordinates {(0.45,0.44) (1,0) (1, 0) (0.45, 0)};

    \draw[-,dashed] (0,0.44) -- (0.45, 0.44);
    \draw[-,dashed] (0,0.285) -- (0.05, 0.285);
    \draw[-,dashed] (0.45,0) -- (0.45, 0.44);
    \draw[-,dashed] (0.05,0) -- (0.05, 0.285);
    \draw[-,dashed] (0.4,0) -- (0.4, 0.3);

    \end{axis}
\end{tikzpicture}

        \caption{
        The instance from the proof of \Cref{lem:pof_rev_upper} with $b = 0.4$. 
        The price of frugality is $0.48 / 0.3 = 1.6$.
        }
        \label{fig:rev_impo}
    \end{subfigure}
    \hfill
    \begin{subfigure}{0.48\textwidth}
        \centering
        \begin{tikzpicture}
    \begin{axis}[
        width=\myheight, height=\myheight,
        xlabel={payment $p$},
        ylabel={},
        xmin=0, xmax=1,
        ymin=0, ymax=1,
        grid=none, 
        legend pos=south east,
        clip=false, 
        xtick={0, 0.16666, 0.333333, 0.5,  1},
        xticklabels={$0$, $0.16$, $0.33$, $0.5$,  $1$},
        ytick={0, 0.2777777, 0.3333333, 0.444444, 0.5, 0.6666666, 1},
        yticklabels={$0$, $0.27$, $0.33$, $0.44$, $0.5$, $0.66$, $1$},
        tick label style={font=\small},
    ]

    \addplot[red, thick] coordinates {(0,0.01) (0.16666,0.01)};
    \addplot[red, thick] coordinates {(0.16666,0.3333333) (0.333333,0.3333333)};
    \addplot[red, thick] coordinates {(0.333333,0.6666666) (0.5,0.6666666)};
    \addplot[red, thick] coordinates {(0.5,1) (1,1)};

    \node at (0.29-0.07, 0.93) {$p \leq b$};
    \node at (0.48-0.07, 0.93) {$p > b$};

    \addplot[red, only marks, mark=*, mark size=1pt] coordinates {(0,0.01) (0.16666,0.3333333) (0.333333, 0.6666666) (0.5, 1) (1,1)};
    \addplot[red, only marks, mark=o, mark size=1pt] coordinates {(0.166666,0.01) (0.333333,0.3333333) (0.5,0.6666666)};

    \addplot[black, thick] coordinates {(0.32,0.0) (0.32,1)};

    \addplot[blue, thick] coordinates {(0,0.0) (0.16666,0.0)};
    \addplot[blue, thick] coordinates {(0.16666,0.2777777) (0.333333,0.2777777)};
    \addplot[blue, thick] coordinates {(0.333333,0.55555555) (0.5,0.55555555)};
    \addplot[blue, thick] coordinates {(0.5,0.833333) (1,0.833333)};
    
    \addplot[blue, only marks, mark=*, mark size=1pt] coordinates {(0,0.0) (0.16666,0.2777777) (0.333333, 0.55555555) (0.5, 0.833333) (1,0.833333)};
    \addplot[blue, only marks, mark=o, mark size=1pt] coordinates {(0.16666,0.0) (0.333333,0.2777777) (0.5,0.55555555)};

    \addplot[green!60!black, thick] coordinates {(0.16666, 0.2777777) (0.333333, 0.2222222)};
    \addplot[green!60!black, thick] coordinates {(0.333333, 0.444444) (0.5, 0.333333)};
    \addplot[green!60!black, thick] coordinates {(0.5, 0.5) (1, 0)};

        \draw[-,dashed] (0, 0.2777777) -- (0.16666, 0.2777777);
        \draw[-,dashed] (0.16666, 0) -- (0.16666, 0.2777777);

    \draw[-,dashed] (0.333333, 0) -- (0.333333, 0.444444);
\draw[-,dashed] (0, 0.444444) -- (0.333333, 0.444444);

 \draw[-,dashed] (0.5, 0) -- (0.5, 0.5);
 \draw[-,dashed] (0, 0.5) -- (0.5, 0.5);

    \addplot[green!60!black, only marks, mark=*, mark size=1pt] coordinates {(0.16666, 0.2777777) (0.333333, 0.444444) (0.5, 0.5) };
    \addplot[green!60!black, only marks, mark=o, mark size=1pt] coordinates {(0.333333, 0.2222222) (0.5, 0.333333) (1, 0)};

    \addplot[fill=red, fill opacity=0.2, draw=none] 
        coordinates {(0.16666,0.3333333) (0.333333,0.3333333) (0.333333, 0.2777777) (0.16666, 0.2777777)};
    \addplot[fill=red, fill opacity=0.2, draw=none] 
        coordinates {(0.333333,0.6666666) (0.5,0.6666666) (0.5, 0.55555555) (0.333333, 0.55555555)};
    \addplot[fill=red, fill opacity=0.2, draw=none] 
        coordinates {(0.5,0.833333) (0.5,1) (1,1) (1,0.833333)};

    \addplot[fill=blue, fill opacity=0.2, draw=none] 
        coordinates {(0.16666,0.2777777) (0.333333,0.2777777) (0.333333, 0.2222222) (0.16666, 0.2777777)};
    \addplot[fill=blue, fill opacity=0.2, draw=none] 
        coordinates {(0.333333,0.55555555) (0.5,0.55555555) (0.5, 0.333333) (0.333333, 0.444444)};
    \addplot[fill=blue, fill opacity=0.2, draw=none] 
        coordinates {(0.5,0.5) (0.5,0.833333) (1, 0.833333) (1, 0)};
    
    \addplot[fill=green!60!black, fill opacity=0.2, draw=none] 
        coordinates {(0.16666,0.2777777) (0.333333,0.2222222) (0.333333, 0) (0.16666, 0)};
    \addplot[fill=green!60!black, fill opacity=0.2, draw=none] 
        coordinates {(0.333333,0.444444) (0.5,0.333333) (0.5, 0) (0.333333, 0)};
    \addplot[fill=green!60!black, fill opacity=0.2, draw=none] 
        coordinates {(0.5,0.5) (1,0) (1, 0) (0.5, 0)};

    \end{axis}
\end{tikzpicture}
        \caption{
        The instance from the proof of \Cref{lem:pof_rev_upper_2} with $b = 1/3$ and $k = 3$.
        The price of frugality is $0.44 / 0.27 \approx 1.8$.
        }
        \label{fig:rev_impo_2}
    \end{subfigure}
    \caption{Illustrations from the proofs of \Cref{lem:pof_rev_upper} and \Cref{lem:pof_rev_upper_2}. In both figures, the red line represents $\maxreward(p)$, the blue line corresponds to $\maxwelfare(p)$, and the green line depicts $(1 - p) \cdot \maxreward(p)$. Notably, for any budget constraint $b$, the maximum value of $(1 - p) \cdot \maxreward(p)$ over $p \in [0,b]$ coincides with $\maxrevenue(b)$. The vertical black line indicates the budget constraint.
    }
    \label{fig:rev_impo_combined}
\end{figure}

\begin{lemma}\label{lem:pof_rev_upper}
For any $0 < b < B \leq 1$ and $0 < \varepsilon < B -b $,
    there exists an 
    instance $\instance$ with additive $f$ and $p(\{i\}) \leq b$ for all $i \in \agents$ such that:
    \begin{align*}
    \gaprevenue(b,B) \geq (1-\varepsilon/2) \cdot (2-b) \xrightarrow{\varepsilon \to 0} 2-b.
    \end{align*}
\end{lemma}
\begin{proof}
The instance used in the following proof is depicted in \Cref{fig:rev_impo}.

    Consider an instance with two agents, where the reward $f$ is additive with $f(\{1\}) = 1/2$ and $f(\{2\}) = 1/2 - b/2$, and the agents' costs are $c_1 = b/2$ and $c_2 = \varepsilon \cdot (1/2 - b/2)^2$. It holds that $p(\{1\}) = c_1 / f_{\{1\}}(1) = b$ and $p(\{2\}) = c_2 / f_{\{2\}}(2) = \varepsilon \cdot (1/2 - b/2)$. Given $\varepsilon < B - b$ and $b < 1$, we obtain $p(\{1,2\}) =c_1 / f_{\{1,2\}}(1) + c_2 / f_{\{1,2\}}(2) = b + \varepsilon \cdot (1/2 - b/2) < B$, implying:
    \begin{align*}
        \maxrevenue(B) &\geq (1-p(\{1,2\})) \cdot f(\{1,2\}) \\
        &= (1-b-\varepsilon \cdot (1/2 - b/2)) \cdot (1-b/2) \\
        &= (1/2 - b/2 - (\varepsilon/2) \cdot (1/2 - b/2)) \cdot (2-b) \\ 
        &= (1-\varepsilon/2) \cdot (1/2 - b/2) \cdot (2-b).
    \end{align*} 
    
    Since $p(\{1,2\}) = b + \varepsilon \cdot (1/2 - b/2) > b$,  we can bound $\maxrevenue(b)$ as follows. 
    Note that $(1-p(\{1\})) \cdot f(\{1\}) = (1-b) \cdot (1/2) = 1/2 - b/2$ and that $(1-p(\{2\})) \cdot f(\{2\}) = (1-\varepsilon \cdot (1/2-b/2)) \cdot (1/2-b/2) < 1/2-b/2$. 
    Thus, $\maxrevenue(b) \leq 1/2 - b/2$, and
    we obtain:
    \begin{align*}
        \frac{\maxrevenue(B)}{\maxrevenue(b)} \geq \frac{(1-\varepsilon/2) \cdot (1/2-b/2) \cdot (2-b)}{1/2-b/2}  = (1-\varepsilon/2) \cdot (2-b),
    \end{align*}
    which concludes the proof.
\end{proof}

We next give our second lower bound on the price of frugality for profit.

\begin{lemma}\label{lem:pof_rev_upper_2}
For any $0 < b < B \leq 1$, any integer $1 \leq k < \min(2B/b, n+1)$,  and $0 < \varepsilon < 2B/k - b$, there exists an 
instance $\instance$ with additive $f$ and $p(\{i\}) \leq b$ for all $i \in \agents$ such that:
    \begin{align*}
    \gaprevenue(b,B) \geq \frac{(2-k \cdot (b + \varepsilon)) \cdot k}{2-b-\varepsilon} \xrightarrow{\varepsilon \to 0} \frac{(2-k \cdot b) \cdot k}{2-b}.
    \end{align*}
Moreover, for any $b$ and $B$, the expression above is minimized at $k = \min(\lfloor 1 / b + 1/2 \rfloor, \lceil 2B/b \rceil - 1, n)$.
\end{lemma}
\begin{proof}
The instance used in the following proof is depicted in \Cref{fig:rev_impo_2}. 

    Consider an instance with $k$ agents, where the reward $f$ is additive with $f(\{i\}) = 1/k$ for all $i \in \agents$ and the costs are $c_i = (1/k) \cdot (b/2 + \varepsilon/2)$ for all $i \in \agents$. It holds that $p(\{i\}) = c_i / f_{\{i\}}(i) = b/2 + \varepsilon/2$. Since $\varepsilon < 2B/k - b$, we get $p(A) = \sum_{i=1}^k c_i / f_{\{i\}}(i) = k \cdot (b/2 + \varepsilon/2) < k \cdot (b/2 + B/k - b/2) = B$. Thus, the optimal profit satisfies $\maxrevenue(B) \geq (1-p(A)) \cdot f(A) = (1 - k \cdot (b/2 + \varepsilon/2)) \cdot 1$.  
    
Since $p(T) = |T| \cdot (b/2 + \varepsilon/2)$, only singletons are budget-feasible. Therefore, we have $\maxrevenue(b) \leq (1 - p(\{1\})) \cdot f(\{1\}) = (1 - b/2 - \varepsilon/2) \cdot (1/k)$. This gives the desired  lower bound on the price of frugality.

For the second part of the statement, note that minimizing $(2-b)/((2-k \cdot b) \cdot k)$ is equivalent to maximizing the quadratic function $(2-k \cdot b) \cdot k = (1/b) \cdot (2/b - k) \cdot k$, which attains its maximum at $k = 1/b$. However, $k = 1/b$ may not be an integer and it may violate $k < \min(2B/b, n+1)$. 

If $\lfloor 1/b + 1/2 \rfloor < \min(2B/b, n+1)$, then the quadratic function is maximized at $k = \lfloor 1/b + 1/2 \rfloor$, as this is the closest integer to $1/b$. Otherwise, if $\lfloor 1/b + 1/2 \rfloor \geq  \min(2B/b, n+1)$, the integer closest to $1/b$ among all permissible values of $k$ is precisely $k = \min(\lceil 2B/b \rceil - 1, n)$.  
\end{proof}

We are now ready to prove \Cref{thm:revenuepofff}.

\begin{proof}[Proof of \Cref{thm:revenuepofff}]
The upper bound for the price of frugality for profit follows from \Cref{lem:revenue_upper}, while the lower bound follows from \Cref{lem:pof_rev_upper,lem:pof_rev_upper_2}.
\end{proof}
\end{document}